\documentclass{llncs}

\usepackage{amsthm}[2004/08/06]
\usepackage[pdftex]{graphicx}
\usepackage[english]{babel}
\usepackage{graphics}
\usepackage{ucs}
\usepackage[utf8x]{inputenc}
\usepackage{url}
\usepackage{ae,aecompl}
\usepackage{centernot}
\usepackage{cite}
\usepackage{amsmath}
\usepackage{amssymb}
\usepackage{fullpage}
\usepackage{datetime}
\usepackage{listings}
\usepackage{graphics,color}
\usepackage[dvipsnames]{xcolor}
\usepackage{multirow}
\usepackage{todonotes}

\graphicspath {{figures/}}

\newcommand{\sm}[1]{{\texttt{#1}}}
\newcommand{\spublic}{\sm{public}}
\newcommand{\sprivate}{\sm{private}}

\newcommand{\tbool}{\sm{bool}}
\newcommand{\tint}{\sm{int}}
\newcommand{\tfloat}{\sm{float}}
\newcommand{\tstring}{\sm{string}}
\newcommand{\pn}[1]{\mathit{#1}}
\newcommand{\pp}{\mathrel{::=}}
\newcommand{\oo}{\mathrel{\mid}}

\newcommand{\sempar}[1]{[\![{{#1}}]\!]}

\newcommand{\sempare}[1]{[\![{{#1}}]\!]^E}
\newcommand{\semparf}[1]{[\![{{#1}}]\!]^F}
\newcommand{\sempars}[1]{[\![{{#1}}]\!]^S}
\newcommand{\semparp}[1]{[\![{{#1}}]\!]^P}

\newcommand{\scln}{\sm{;}\ }

\newcommand{\screturn}[1]{{\sm{return}\ {#1}}}
\newcommand{\scdecl}[2]{\sm{\{} {#1} \scln {#2} \sm{\}}}
\newcommand{\sccall}[2]{{{#1}\sm{(}{#2}\sm{)}}}
\newcommand{\scassign}[2]{{#1}\ \sm{=}\ {#2}}

\newcommand{\scskip}{{\sm{skip}}}

\newcommand{\dle}[1]{{{#1}}^{PL}}
\newcommand{\sce}[1]{{{#1}}^{SC}}

\newcommand{\sempardl}[1]{\dle{\sempar{{#1}}}}
\newcommand{\semparsc}[1]{\sce{\sempar{{#1}}}}

\newcommand{\NN}{\mathbb{N}}

\newcommand{\set}[1]{\{{#1}\}}

\newcommand{\tm}[1]{\text{\texttt{{#1}}}}
\newcommand{\nt}[1]{\mathit{{#1}}}
\newcommand{\kleene}[1]{{#1}^{*}}
\newcommand{\kleenepos}[1]{{#1}^{+}}

\newcommand{\argument}{\tm{argument}}
\newcommand{\getTable}{\tm{getTable}}
\newcommand{\declassify}{\tm{declassify}}
\newcommand{\publish}{\tm{publish}}
\newcommand{\shuffle}{\tm{shuffle}}
\newcommand{\concat}{\tm{concat}}
\newcommand{\first}{\tm{markFirst}}
\newcommand{\unique}{\tm{unique}}
\newcommand{\filter}{\tm{filter}}
\newcommand{\join}{\tm{crossPr}}

\newcommand{\lbl}{\tm{label}}
\newcommand{\B}{\mathcal{B}}
\newcommand{\LL}{\mathcal{C}}

\newcommand{\transfun}{T}
\newcommand{\transfunp}{\transfun^{P}}
\newcommand{\transfunf}{\transfun^{F}}
\newcommand{\transfunc}{\transfun^{C}}

\newcommand{\transp}[1]{\transfunp({#1})}
\newcommand{\transf}[1]{\transfunf({#1})}
\newcommand{\transc}[1]{\transfunc({#1})}

\newcommand{\stransfun}{S}
\newcommand{\stransfunp}{\stransfun^{P}}
\newcommand{\stransfunf}{\stransfun^{F}}
\newcommand{\stransfunc}{\stransfun^{C}}
\newcommand{\stransfung}{\stransfun^{G}}
\newcommand{\stransfune}{\stransfun^{T}}

\newcommand{\stransp}[1]{\stransfunp({#1})}
\newcommand{\stransf}[1]{\stransfunf({#1})}
\newcommand{\stransc}[1]{\stransfunc({#1})}
\newcommand{\stransg}[1]{\stransfung({#1})}
\newcommand{\stranse}[1]{\stransfune({#1})}
\newcommand{\stransop}[1]{{#1}'}

\newcommand{\adornfun}{A}
\newcommand{\adornfunp}{\adornfun^{P}}
\newcommand{\adornfunf}{\adornfun^{F}}
\newcommand{\adornfunc}{\adornfun^{C}}
\newcommand{\adornfung}{\adornfun^{G}}

\newcommand{\adornp}[1]{\adornfunp({#1})}
\newcommand{\adornf}[2]{\adornfunf_{{#1}}({#2})}
\newcommand{\adornc}[2]{\adornfunc_{{#1}}({#2})}
\newcommand{\adorng}[2]{\adornfung_{{#1}}({#2})}

\newcommand{\ico}{\mathcal{I}}
\newcommand{\sico}{\hat{\mathcal{I}}}

\newcommand{\lpneg}{\tm{\textbackslash+}}
\newcommand{\lpand}{\tm{,}}
\newcommand{\lpor}{\tm{;}}
\newcommand{\imp}{\tm{:-}}

\newcommand{\lpend}{\tm{.}}

\newcommand{\goal}{\tm{goal}}

\newcommand{\lppow}{\tm{\^{}}}

\newcommand{\pred}{\mathsf{pred}}

\newcommand{\indnt}{\tm{\ \ \ \ }}

\newcommand{\choice}[1]{\mathsf{choose}_{{#1}}}
\newcommand{\declared}{\mathsf{decl}}

\newcommand{\vars}{\mathsf{vars}}
\newcommand{\fvars}{\mathsf{fvars}}
\newcommand{\schema}{\mathsf{schema}}

\newcommand{\size}{\mathsf{size}}
\newcommand{\arity}{\mathsf{arity}}

\newcommand{\infrule}[3][{}]{$\begin{array}{c} {#3} \\ \hline {#2}  \end{array}{#1}$}

\newcommand{\clauseset}{Set({\nt{clause}})}
\newcommand{\emptystate}{\emptyset}

\newcommand{\dom}{d}
\newcommand{\type}{t}

\newcommand{\R}{\mathcal{R}}
\newcommand{\RB}{\mathcal{RB}}

\newcommand{\univ}{\mathcal{U}}

\newcommand{\rel}{\mathbf{Rel}}

\newcommand{\vecset}[1]{\vals[{#1}]}

\newcommand{\derwrt}[1]{\stackrel{{#1}}{\sim}}
\newcommand{\tables}{\mathcal{T}}
\newcommand{\funset}[1]{\mathcal{F}_{{#1}}}

\newcommand{\vals}{\mathbf{Val}}
\newcommand{\vnames}{\mathbf{VName}}

\newcommand{\funs}{\rho}
\newcommand{\eqvs}[1]{\sim_{{#1}}}

\newcommand{\share}[1]{\langle\!\langle{#1}\rangle\!\rangle}

\renewcommand{\vec}[1]{\mathbf{#1}}
\newcommand{\mat}[1]{\mathbf{#1}}
\newcommand{\prot}[1]{\mathsf{#1}}

\newcommand{\relations}[1]{\rel_{{#1}}}

\title{PrivaLog: a privacy-aware logic programming language}

\author{Joosep J\"{a}\"{a}ger\inst{1,2} \and Alisa Pankova\inst{1}}
\institute{Cybernetica, Estonia\\
\email{\{joosep.jaager,alisa.pankova\}@cyber.ee}
}

\begin{document}

\maketitle

\begin{abstract}
Logic Programming (LP) is a subcategory of declarative programming that is considered to be relatively simple for non-programmers. LP developers focus on describing facts and rules of a logical derivation, and do not need to think about the algorithms actually implementing the derivation.

Secure multiparty computation (MPC) is a cryptographic technology that allows to perform computation on private data without actually seeing the data. In this paper, we bring together the notions of MPC and LP, allowing users to write privacy-preserving applications in logic programming language.
\end{abstract}

\section{Introduction}
\label{sec:intro}

Logic Programming (LP) is a simple, yet powerful formalism suitable for programming and knowledge representation~\cite{DBLP:books/daglib/0096993}. A typical logic program is represented as a set of facts and rules that derive new facts from the existing ones. Logic programming engine will use given facts and rules to perform logical reasoning, using internal mechanisms that the programmer does not need to describe. Program output can be either a single boolean value denoting whether a certain predicate (called a \emph{goal}) can be satisfied, or a set of variable valuations satisfying the goal.

Secure multiparty computation (MPC) is a cryptographic technology that allows to perform computation on data without actually seeing the data~\cite{UT:Kamm15}. A standard use case of MPC is when several distinct parties (e.g. companies) want to compute a function that includes all their data, but they do not want to disclose that data to each other. Secure MPC allows to perform the computation in such a way that the parties do not learn anything besides the final output, e.g. some aggregated statistics.

MPC applications are typically written in imperative languages~\cite{DBLP:conf/ecoop/BogdanovLR14, liu2015oblivm, fairplaymp, oblivc, picco, aby, DBLP:conf/pldi/NielsenS07} or functional languages~\cite{DBLP:conf/csfw/Mitchell0SZ12, wysteria, DBLP:conf/sp/LiuHSKH14}. A privacy-preserving programming language based on Datalog has been proposed in~\cite{lpsec1}. That short paper describes a general framework for transforming a Datalog program (with certain constraints) to a privacy-preserving application. In particular, it outputs a program that can be run on Sharemind~\cite{bogdanov2008sharemind} framework.

Sharemind is one of numerous frameworks that support deployments of MPC using special API call interfaces~\cite{archer2018keys}. The main protocol set~\cite{DBLP:journals/ijisec/BogdanovNTW12} of Sharemind is based on secret-sharing among three parties. The client who makes a query to a private database and receives the output can be a separate entity. Sharemind has numerous subprotocols available for use, e.g. oblivious shuffle~\cite{Laur:2011:ROD:2051002.2051027}. In this work, we do not invent yet another MPC platform, but provide a translator from a logic program to a Sharemind application. Our work is based on~\cite{lpsec1}, where the source logic program is translated to SecreC language~\cite{DBLP:conf/ecoop/BogdanovLR14}. SecreC is an imperative language with privacy type inference, which does not allow accidental private data leakage.

\paragraph{Contribution of this paper}

In this work, we extend and deepen the framework proposed in~\cite{lpsec1}. In particular, we do the following.
\begin{itemize}
\item Describe full end-to-end transformation from the source Datalog program to SecreC program.
\item Introduce some optimizations.
\item Implement the translator and evaluate its performance.
\end{itemize}
The translator is available for public use in a GitHub repository, and it is compatible with both the licensed Sharemind, as well as freely available Sharemind Emulator that does not require a license.

\section{Preliminaries}\label{sec:preliminaries}

\subsection{Notation}


Syntax of logic programs may depend on a particular language. Throughout this paper, we will be using the notation described in Table~\ref{tbl:syntax}. The notation for arithmetic operations and inequalities is quite standard and is not given in the table.

\begin{table*}[t]
\centering
\caption{Basic elements of a logic programming language}\label{tbl:syntax}
\begin{tabular}{| l | l | l |}
\hline
element & details & example \\
\hline
\hline
variable & name that starts with an upper case letter & X \\
\hline
unnamed  & underscore & \texttt{\_} \\
variable &  & \\
\hline
constant & name that starts with a lower case letter & ship \\
         & a number & 12 \\
\hline
logical  & AND & \texttt{p \lpand\ q} \\
operator & OR & \texttt{p \lpor\ q} \\
         & NOT & \texttt{\lpneg p} \\
\hline
comparison  & unification & \texttt{X = Y} \\
            & arithmetic equality & \texttt{X =:= Y} \\
            & arithmetic inequality & \texttt{X =/= Y} \\
\hline
fact & a predicate declaration & \texttt{p(a1,\ldots,an) \lpend} \\
\hline
rule & \texttt{p} is true if \texttt{q} is true          & \texttt{p \imp\ q \lpend} \\
\hline
goal & the query that produces the final output & \texttt{?- p(x1,\ldots,xn)\lpend} \\
\hline
\end{tabular}
\end{table*}


\subsection{Running example}\label{sec:example}

A logic program is a list of facts, and rules describing how new facts can be derived from the existing ones. The facts can be viewed as logical predicates, or as database relations. Let us start from a small example, inspired from~\cite{DBLP:conf/fase/TootsTYDGLMPPPT19}, which will be used as a running example throughout this paper.

Datalog distinguishes between \emph{extensional} predicate symbols (defined by facts) and \emph{intensional} predicate symbols (defined by rules). The corresponding database relations are denoted EDB and IDB respectively. Our example program has two EDB relations:
\begin{itemize}
\item \texttt{ship(name,latitude,longitude,speed,cargo,amount)} tells that there is a ship with a certain \texttt{name}, located at (\texttt{latitude},\texttt{longitude}), which moves with certain \texttt{speed} and carries a certain \texttt{amount} of \texttt{cargo}.
\item \texttt{port(name,latitude,longitude,capacity)} tells that a port with a certain \texttt{name} is located at (\texttt{latitude},\texttt{longitude}), which can accommodate as much cargo as fits into its \texttt{capacity}.
\end{itemize}
The logic program starts with facts corresponding to these two relations. Each fact ends with a period.
\begin{small}
\begin{verbatim}
ship(alfa, 270, 290, 40, onions, 10).
ship(beta, 180, 280, 30, garlic, 5).
...
port(alma, 0, 0, 10).
port(milka, 10, 10, 10).
...
\end{verbatim}
\end{small}

We now need to list the rules. The rules describe how IDB relations are derived from EDB relations, as well as the other IDB relations.

The IDB relation \texttt{reachability\_time(Ship,Port,Time)} is true if the \texttt{Ship} can reach the \texttt{Port} in certain \texttt{Time}.

\begin{small}
\begin{verbatim}
reachability_time(Ship,Port,Time) :-
    ship(Ship,X1,Y1,Speed,_,_),
    port(Port,X2,Y2,_),
    Time = sqrt((X1 - X2)^2 + (Y1 - Y2)^2) / Speed.
\end{verbatim}
\end{small}

The IDB relation \texttt{suitable\_port(Ship,Port)} is true if the \texttt{Port} can accommodate cargo of the \texttt{Ship}.

\begin{small}
\begin{verbatim}
suitable_port(Ship,Port) :-
    port(Port,_,_,Capacity),
    ship(Ship,_,_,_,_,CargoAmount),
    Capacity >= CargoAmount.
\end{verbatim}
\end{small}

The IDB relation \texttt{arrival(Ship,Port,CargoType,Time)} depends on the previous two relations, and it is true if the \texttt{Ship} carrying \texttt{CargoType} can reach a suitable \texttt{Port} within certain \texttt{Time}.
\begin{small}
\begin{verbatim}
arrival(Ship,Port,CargoType,Time) :-
    ship(Ship,_,_,_,CargoType,_),
    suitable_port(Ship,Port),
    reachability_time(Ship,Port,Time).
\end{verbatim}
\end{small}

Finally, there is a goal querying the relation \texttt{arrival} with certain \texttt{portname} and \texttt{cargotype}.
\begin{small}
\begin{verbatim}
?-arrival(Ship,portname,cargotype,Time).
\end{verbatim}
\end{small}
The program will output all pairs (\texttt{Ship},\texttt{Time}) where \texttt{Ship} carries \texttt{cargotype} and can reach the \texttt{port} within the \texttt{Time}. The names \texttt{cargotype} and \texttt{port} serve as placeholders for actual values that are treated as program input.

\subsection{Relational algebra}\label{sec:prelim:relalg}

In this paper, we will formalize logic programs treating logic predicates as relations. We follow the description of relational algebra from~\cite{DBLP:books/sp/CeriGT90}, where it is defined in logic programming context. In the relational model, data are organized using relations, which are defined as follows.
\begin{definition}[relation]
Let $D_1,\ldots, D_n$ be sets, called domains (not necessarily distinct). Let $D = D_1 \times \ldots \times D_n$ be the Cartesian product. A \emph{relation} defined on $D$ is any subset of $D$.
\end{definition}

We suppose domains to be finite sets of data. Unless stated otherwise, we will therefore assume finite relations.
\begin{definition}[tuple of a relation over domain]
A \emph{tuple of a relation} $\R$ over $D = D_1 \times \cdots \times D_n$ is a tuple $(d_1,\ldots,d_n)$, with $d_1 \in D_1,\ldots, d_n \in D_n$.
\end{definition}
The number of tuples in $\R$ is called the cardinality of $\R$. The empty set is a particular relation, the null (or empty) relation. Note that a relation is a set, therefore tuples in a relation are distinct, and the order of tuples is irrelevant. However, the order of elements in each tuple \emph{is} relevant, and each position has a special name called an attribute of relation. Let $\vnames$ be the set of variable names.
\begin{definition}[schema]
A \emph{schema} of a relation $\R$ is an ordered set of all its attributes $\set{a_1,\ldots,a_n} \subseteq \vnames^n$.
\end{definition}
We may refer to attributes of relations either by their name (denoted $\R[x]$, where $x \in \vnames$) or by the position (denoted $\R[i]$, where $i \in \NN$). Let $\univ_{\Gamma}$ denote the a universal relation (i.e. a relation that contains all possible tuples) with schema $\Gamma$. Let $\relations{\Gamma}$ be the set of all possible relations over schema $\Gamma$, i.e. the set of all subsets of $\univ_{\Gamma}$. Let $\emptyset_{\Gamma}$ denote an empty relation over schema $\Gamma$.

Relational Algebra is defined through several operators that apply to relations, and produce other relations. A relational algebra operator can be a selection (filter), projection, a cartesian product, union, intersection, difference, and renaming of columns. The semantics of these operators w.r.t. relation sets is given in Figure~\ref{fig:relops}. A filter is using a boolean expression $q: \R \to \set{true,false}$, which formally maps each row $r \in \R$ of a relation $\R$ to a boolean value. its semantics are given in Figure~\ref{fig:boolsemantics}, where function $\phi$ defines semantics of predicates not covered by propositional logic, e.g. inequalities.

\begin{figure}
\begin{align*}
\sempar{false}(r) &= false\\
\sempar{true}(r) &= true\\
\sempar{f(s_1,\ldots,s_k)}(r) &= \phi(f)(\sempar{s_1}\ r, \ldots, \sempar{s_k}\ r)\\
\sempar{q_1\wedge q_2}(r) &= \sempar{q_1}(r) \wedge \sempar{q_2}(r)\\
\sempar{q_1\vee q_2}(r) &= \sempar{q_1}(r) \vee \sempar{q_2}(r)\\
\sempar{\neg q}(r) &= \neg\sempar{q}(r)
\end{align*}
\caption{Semantics of boolean expressions}\label{fig:boolsemantics}
\end{figure}

\begin{figure}
\begin{align*}
\pi_{\Gamma}\ \R & = \bigcup_{\substack{\R' \in \univ_{\Gamma}\\ \forall x\in \Gamma: \R[x]=\R'[x]}} \R' & & \text{undefined, if } \Gamma\not\subseteq\schema(\R)\\
\sigma_{q}\ \R & = \set{r \in \R\ |\ \sempar{q}\ r = true} & & \text{undefined, if } \vars(q)\not\subseteq\schema(\R)\\
\R_1 \times \R_2 & =\set{r_1 \| r_2\ | r_1 \in \R_1, r_2 \in \R_2}  & & \text{undefined, if } \schema(\R_1)\cap\schema(\R_2)\not=\emptyset\\
\R_1 \cap \R_2 & = \R_1 \cap \R_2 \text{ for }\schema(\R_1) = \schema(\R_2) & & \text{undefined, if } \schema(\R_1)\not=\schema(\R_2)\\
\R_1 \cup \R_2 & = \R_1 \cup \R_2 \text{ for }\schema(\R_1) = \schema(\R_2) & & \text{undefined, if } \schema(\R_1)\not=\schema(\R_2)\\
\R_1 \setminus \R_2 & = \R_1 \setminus \R_2 \text{ for }\schema(\R_1) = \schema(\R_2) & & \text{undefined, if } \schema(\R_1)\not=\schema(\R_2)\\
\end{align*}
\caption{Semantics of relational operators w.r.t. relation sets}\label{fig:relops}
\end{figure}

When applying operators, we treat relations as set, and it may be unclear how the resulting schema looks like. We define schemas of relational expressions in~\ref{fig:relschema}.

\begin{figure}
\begin{align*}
\schema(\pi_{\Gamma}(\R)) & = \Gamma & & \text{undefined, if } \Gamma\not\subseteq\schema(\R)\\
\schema(\sigma_q(\R)) & = \schema(\R) & & \text{undefined, if } \vars(q)\not\subseteq\schema(\R)\\
\schema(\R_1\times \R_2) &= \schema(\R_1)\|\schema(\R_2) & & \text{undefined, if } \schema(\R_1)\cap\schema(\R_2)\not=\emptyset\\
\schema(\R_1\cup \R_2) & = \schema(\R_1) & & \text{undefined, if } \schema(\R_1)\not=\schema(\R_2)\\
\schema(\R_1\cap \R_2) & = \schema(\R_1) & & \text{undefined, if } \schema(\R_1)\not=\schema(\R_2)\\
\schema(\R_1 \setminus \R_2) & = \schema(\R_1) & & \text{undefined, if } \schema(\R_1)\not=\schema(\R_2)\\
\end{align*}
\caption{Semantics of relational operators w.r.t. relation schemas}\label{fig:relschema}
\end{figure}

In this work, we will use the following properties of relational algebra (without proofs):
\begin{proposition}[properties of relational algebra]\label{prop:relalgebra}
Let $\R$ be a relation. The following holds:
\begin{itemize}
\item $\sigma_{\neg q}(\R) = \R\ \setminus\ \sigma_{q}(\R)$;
\item $\sigma_{q_1 \wedge q_2}(\R) = \sigma_{q_1}(\R) \cap \sigma_{q_2}(\R)$;
\item $\sigma_{q_1 \vee q_2}(\R) = \sigma_{q_1}(\R) \cup \sigma_{q_2}(\R)$;
\item $\sigma_{q}(\R) = \sigma_{q}(\univ_{\Gamma}) \cap \R$ for $\schema(\R) \subseteq \Gamma$.
\end{itemize}
\end{proposition}

\subsection{Deployment Model and the SecreC Language}\label{sec:secrec}

Sharemind platform performs computation on secret-shared data. Technically, the platform is installed on three separate servers. All data is shared among these servers in such a way that every server sees its part of the data as a random bitstring and cannot infer what the data actually is. This property is preserved during the computation, where specially designed protocols convert secret-shared inputs to secret-shared outputs. For this paper, it is not relevant how exactly the data is shared and how the underlying protocols work.

In this paper, we consider the following deployment model of Sharemind system:
\begin{itemize}
\item First of all, the are \emph{input parties} who provide the data. The data is stored in Sharemind platform in a secret-shared manner and is not visible to anyone.
\item There are three designated \emph{computing parties} (usually different from the input parties) who will perform the computation on these data, without actually seeing the data.
\item Finally, there is a \emph{client} (who can be one of the input parties, or a separate entity) who receives the output of the computation.
\end{itemize}
More information about deployment models of Sharemind can be found in~\cite[Table 2.2]{UT:Kamm15}.

The computation performed by computing parties is described in SecreC language. SecreC is a C-like programming language designed for writing privacy-preserving applications for the Sharemind platform~\cite{UT:Randmets17,secrec, KLR16}. SecreC distinguishes between public data (visible to all computing parties) and private data (not visible to any computing party) on type system level. The only way to convert private data to public (which may be needed in some applications) is to use an explicit {\declassify} command, which works as decryption. The type system of SecreC ensures non-interference between public and private variables if either:
\begin{itemize}
\item there is no {\declassify} operator in the SecreC program;
\item the {\declassify} operator is used as a part of some semantic function (like sorting) for which non-interference has been proven.
\end{itemize}

In this work, we translate LP programs to SecreC language. SecreC compiler can then be used to perform the type check (and thus ensure that private data is not leaked during the computation), and transform a SecreC program to a sequence of Sharemind API calls.

When writing a program in SecreC, one should take into account that SecreC does not allow control flow to depend on private data. Let us bring a simple example, taken from ~\cite{lpsec1}. Suppose that there is a condition \texttt{S < 2500} where \texttt{S} has private type. SecreC does not allow to check directly whether the inequality is true, as it would leak some information about the secret \texttt{S}. Instead, SecreC allows to compute a \emph{private} bit $b \in \set{0,1}$, denoting whether the inequality is true. The programmer needs to restructure the subsequent computation, which in general requires simulation of computation in \emph{both} branches (the cases \texttt{S < 2500} and \texttt{S >= 2500}), and using the bit $b$ to choose the final output for each variable affected by computation in these branches. For example, if a variable $y$ gets a value $y_0$ in the case $b = 0$, and a value $y_1$ in the case $b = 1$, we need to compute both $y_0$ and $y_1$, and eventually take $y = (1 - b) \cdot y_0 + b \cdot y_1$. Since $b$ has private type, the variable $y$ is also assigned private type, so the computing parties will not know whether $y = y_0$ or $y = y_1$. In context of LP, this means that we cannot immediately backtrack when reaching a false statement that depends on private data. Instead, for each potential solution, we get a private satisfiability bit $b$, and in the end the client will receive only those solutions for which $b = 1$.

Since even computation of simple operations like multiplication or comparison on private data requires network communication between computing parties, SecreC programs are much more efficient if the data are processed in larger batches using SIMD (Single Instruction Multiple Data) operations. For example, in~\cite{secrec,fim}, an iterative SecreC program that constructs frequently occurring itemsets of certain size, and the same privacy-preserving operations are applied at once in all search branches for the entire depth level of all constructed itemsets of size $k$. Hence, the programmer should use as many vectorized computations as possible. For that reason, in this work, we will process different possible non-deterministic computational branches simultaneously in parallel.

\paragraph{Existing SecreC functions.} In this work, we will not invent new secure MPC protocols, but instead build our solution of top of numerous existing Sharemind protocols already available. In particular, SecreC already has the following functionalities (which we denote $\funset{}$):
\begin{itemize}
\item \textbf{Arithmetic Blackbox Functions:} SecreC has numerous built-in functions for privacy-preserving arithmetic and boolean operations.
\item \textbf{Existing Protocols:} SecreC has some protocols for frequently used tasks (like shuffle and sorting) already implemented.
\item \textbf{Database operations:} SecreC allows to store data in a private database and access it during the computation.
\end{itemize}
We list some existing SecreC functions that are relevant for this paper in Table~\ref{tbl:bb}. We use notation $\share{x}$ to denote a secret-shared value $x$.
\begin{table*}[t]
\begin{center}
\begin{tabular}{| c | c |}
\hline
Operation Call & Effect \\ 
\hline
\hline
$\getTable(t)$ & return columns of a table $t$ from a private database\\
\hline
$\join(\share{\mat{X_1}},\ldots,\share{\mat{X_n}})$ & return a cross product $\share{\mat{X_1} \times \cdots \times \mat{X_n}}$ \\
\hline
$\declassify(\share{\vec{x}})$ & decrypt $\share{\vec{x}}$ and return $\vec{x}$ \\
\hline
$\argument(i)$ & receive the $i$-th input from the client \\
\hline
$\publish(\share{\vec{x}})$ & send output $\share{\vec{x}}$ to the client (who decrypts it locally) \\
\hline
$\first(\share{\vec{x}})$ & return a boolean vector $\share{\vec{b}}$ s.t. $b_i=1$ iff it is the first occurrence of $x_i$ in $\vec{x}$\\
\hline
$\shuffle(\share{\vec{x}})$ & return $\share{\vec{y}}$, where $\vec{y}$ a random reordering of $\vec{x}$ \\
\hline
$\concat(\share{\vec{x_1}},\ldots,\share{\vec{x_n}})$ & return $\share{\vec{x_1}} \| \cdots \| \share{\vec{x_n}}$\\
\hline
$\filter(\share{\vec{x}}, \vec{b})$ & return $\share{[x_i \in \vec{x}\ |\ b_i = true]}$\\
\hline
\end{tabular}
\caption{SecreC built-in functions $\funset{}$}\label{tbl:bb}
\end{center}
\end{table*}

\paragraph{Syntax of SecreC} The syntax of SecreC core language, taken from~\cite{UT:Randmets17}, is given in Figure~\ref{fig:secrec:syntax}. Compared to original core syntax, we remove conditional statements and while-loops, since we will not be using these constructions in our SecreC programs. We also leave out declarations of \emph{protection domain}, which is not relevant for this work. The protection domain defines used secret-sharing scheme (assumed attacker strength, number of computing parties etc), and we just assume that this scheme has been fixed in advance and remains implicit. We also assume that return statement is always the last line of function body. 

\begin{figure}
    \centering
    \begin{minipage}[t]{0.8\linewidth}
    \centering
\[
    \begin{array}{rcl}
        \pn{prog} & \pp & \kleene{\pn{fun}}\ \ \tm{main()}\ \pn{stmt} \\
        \pn{fun} & \pp & f(x_1 : d_1\ t_1, \ldots, x_n : d_n\ t_n) : (d'_1\ t'_1,\ldots,d'_m\ t'_m)\ \\ & & (stmt ; \screturn{e_1,\ldots,e_m})
    \end{array}
\]
    \end{minipage}\\

    \begin{minipage}[t]{0.45\linewidth}
        \centering
\[
    \begin{array}{rcl}
     stmt & \pp & \scskip \\
          & \oo & stmt_1\ \sm{;}\ stmt_2 \\
          & \oo & \scassign{x}{e} \\
          & \oo & \scdecl{x : d\ t}{stmt} \\
    \end{array}
\]
    \end{minipage}
    \begin{minipage}[t]{0.45\linewidth}
        \centering
\[
    \begin{array}{rcl}
       e & \pp & x
          \oo c^t
          \oo !e \\
          & \oo & e_1 \oplus e_2 \text{ for }\oplus \in\set{\tm{+},\tm{-},\tm{*},\tm{/},\lppow,\tm{\&},\tm{|}}\\
          & \oo & e_1 \odot e_2 \text{ for }\odot \in\set{\tm{<},\tm{<=},\tm{==},\tm{!=},\tm{>=},\tm{>}}\\
          & \oo & \sccall{f}{e_1, e_2, \ldots, e_n} \\
    \end{array}
\]
    \end{minipage}
    \begin{minipage}[t]{0.45\linewidth}
        \centering
\[
    \begin{array}{rcl}
       f & \pp & \text{function name}\\
       x & \pp & \text{variable name}\\
       c^t & \pp & \text{constant of type }t
    \end{array}
\]
    \end{minipage}
    \begin{minipage}[t]{0.45\linewidth}
        \centering
\[
    \begin{array}{rcl}
        d & \pp & \spublic \oo \sprivate\\
        t & \pp & \tint \oo \tbool \oo \tfloat \oo \tstring \\ & & \tint\sm{[]} \oo \tbool\sm{[]}\oo \tfloat\sm{[]} \oo \tstring\sm{[]}

    \end{array}
\]
    \end{minipage}

    \caption{Abstract syntax of the SecreC core language. Here $\kleene{z}$ denotes repetition of $z$ zero or more times.}
    \label{fig:secrec:syntax}
\end{figure}

\paragraph{Semantics of SecreC} 
Semantics of SecreC is similar to one of a statically typed imperative programming language. We give simplified denotational semantics of a SecreC program in Figure~\ref{fig:secrec:semantics}, which is sufficient for the programs that will be constructed in this paper.

Generated SecreC programs will work on vectors, where a vector corresponds to a certain program variable, and each element of the vector corresponds to one potential valuation of that variable, i.e. one non-deterministic solution. Hence, a program state is a map $s: \vnames \to \vecset{m}$, where $\vnames$ is the set of variable names, and $\vecset{m}$ is a set of vectors of length $m$, where $m$ is the same for all variables. The value $m$ depends on the sizes of database tables, as well as the number of potential solutions, and it may change throughout execution of a SecreC program. We treat database $\tables : \vnames \to \vals^{n_1 \times n_2}$ as a special input of a SecreC program.


\begin{figure}
    \centering
    \begin{eqnarray*}
    \sempare{x}\ \funs\ s & = & s[x]\\
    \sempare{c^t}\ \funs\ s & = & c^t\\
    \sempare{!e}\ \funs\ s & = & \neg\sempare{e}\ \funs\ s\\
    \sempare{e_1 \oplus e_2}\ \funs\ s & = & (\sempare{e_1}\ \funs\ s)\ \oplus\ (\sempare{e_2}\ \funs\ s)\\
    & & \text{ for }\oplus \in\set{\tm{+},\tm{-},\tm{*},\tm{/},\lppow,\tm{\&},\tm{|}}\\
    \sempare{e_1 \odot e_2}\ \funs\ s & = & (\sempare{e_1}\ \funs\ s)\ \odot\ (\sempare{e_2}\ \funs\ s)\\
    & & \text{ for }\odot \in\set{\tm{<},\tm{<=},\tm{==},\tm{!=},\tm{>=},\tm{>}}\\
    \sempare{f(e_1,\ldots,e_n)}\ \funs\ s & = & (\funs\,f) (\sempare{e_1}\ \funs\ s,\ldots, \sempare{e_n}\ \funs\ s)
    \end{eqnarray*}
    \begin{eqnarray*}
    \sempars{skip}\ \funs\ s & = & s \\
    \sempars{stmt_1\ ;\ stmt_2}\ \funs\ s & = & (\sempars{stmt_1}\ \funs \circ \sempars{stmt_2}\ \funs)\ s \\
    \sempars{x = e}\ \funs\ s & = & s[x \mapsto \sempare{e}\ \funs\ s]\\
    \sempars{\scdecl{x : d\ t}{stmt}}\ \funs\ s & = & \sempars{stmt}\ s
    \end{eqnarray*}
    \begin{equation*}
    \semparf{f(x_1, \ldots, x_n)\ (stmt\ ;\ \screturn{y_1,\ldots,y_m})}\ \funs\ [v_1,\ldots,v_n] = (\sempars{stmt}\ \funs\ \set{x_1 \mapsto v_1,\ldots,x_k \mapsto v_k})\ [y_1,\ldots,y_m]
    \end{equation*}\\
    \begin{eqnarray*}
    \semparp{f_1\ \ldots\ f_m\ stmt}\ (v_1, \ldots, v_n)\ \tables &=& \sempars{stmt}\ \funs_m\ \emptystate\\
    \text{where }\funs_0 & = & \bigcup_{p \in Domain(\tables)} \set{{\getTable}(p) \mapsto \tables(p)} \cup \bigcup_{i=1}^n \set{\prot{argument}(i) \mapsto v_i} \cup \funset{}\\
                 \funs_i & = & \set{f_i \mapsto \semparf{f_i}\ \funs_{i-1}}\text{ for }i \in \set{1,\ldots,m}
    \end{eqnarray*}
    \caption{Denotational semantics of the SecreC core language. Here $s: \vnames \to \vecset{m}$ is a program state, $\emptystate$ denotes an empty state, $\tables : \vnames \to \vals^{n_1 \times n_2}$ the private database, and $\funs : f \mapsto (\vals \to \vals)$ defines semantics of functional symbols, which includes build-in functions and functions declared inside the program.}
    \label{fig:secrec:semantics}
\end{figure}

\section{The PrivaLog language}

In this work, we are dealing with logic programs like the one described in Sec.~\ref{sec:example}. We enhance the input data (including the database) with privacy labels (\tm{public} and \tm{private}) to define which data need to remain private throughout the computation. Putting aside privacy labels (which are a new component), our input language is a subset of Datalog. Due to new constructions and differences in syntax, we treat it as a different language, which we call PrivaLog.

In Datalog, it would be natural to list the database relations as facts inside the logic program itself. However, private data cannot occur in plaintext as a part of the program. Instead, we can only list the structure of database predicates and the types of their arguments. The data will be uploaded to MPC platform by input parties separately, using special secure data upload mechanisms. Facts that contain only public data can still be defined inside program code.

To convert the running example of Sec.~\ref{sec:example} to PrivaLog, we need to label all inputs (including attributes EDB predicates) with privacy labels. Since SecreC is a typed language, we need to introduce data types into PrivaLog program as well. Instead of writing \texttt{ship(alma,270,290,onions,10)} etc, we just write schema declarations. We add types to the inputs of the goal. Let us show how it changes our running example.
\begin{small}
\begin{verbatim}
:-type(ship(name        : public string,
            latitude    : private float,
            longitude   : private float,
            speed       : public int,
            cargotype   : private string,
            cargoamount : private int)).

:-type(port(name            : public string,
            latitude        : public float,
            longitude       : public float,
            offloadcapacity : public int,
            available       : public bool)).
....
?-arrival(Ship, portname : private string,
          cargotype : private string, Time).
\end{verbatim}
\end{small}

\subsection{Syntax}\label{sec:syntax}

The syntax of PrivaLog is given in Figure~\ref{fig:syntax}. The schema of extensional database is defined in the beginning of a program. The attributes of EDB relations are typed. We need to specify whether an attribute is public or private, as well as its data type, which determines how exactly private data will be handled technically inside Sharemind (sharing is different e.g. for integers and strings). W.l.o.g. we assume that the goal consists of a single atomic predicate, as we can fold a more complex goal into a separate rule if needed. The names $a_1,\ldots,a_n$ in the goal statement serve as placeholders for program inputs, and will be replaced by actual values when the program is executed. While syntax shows that $a_i$ are listed before $x_i$ in the goal predicate, this is only for simplicity of representation, and in practice their order may be arbitrary.
\begin{figure}
\begin{small}
\begin{eqnarray*}
\nt{program} & ::= & \kleenepos{\nt{schema}}\ \kleenepos{\nt{clause}}\ \nt{goal}\\
\nt{schema} & ::= & \imp\ \tm{type}(p(a_1 : d_1\ t_1,\ldots,a_n\ :\ d_n\ t_n)) \lpend\\
\nt{d} & ::= &  \tm{public}\ |\ \tm{private}\\
\nt{t} & ::= &  \tm{int}\ |\ \tm{bool}\ |\ \tm{float}\ | \tm{string}\\
\nt{goal} & ::= & \tm{?-}\ p(a_1\ :\ d_1\ t_1,\ldots,a_n\ :\ d_n\ t_n,x_1,\ldots,x_m)\ \lpend\\
\nt{clause} & ::= & \nt{fact}\ |\ \nt{rule}\\
\\
\nt{fact} & ::= &  p(c^{t_1}_1,\ldots,c^{t_n}_n) \lpend\\
\nt{rule} & ::= &  p(x_1,\ldots,x_n) \imp\ \nt{formula}\ \lpend\\
\nt{formula} & ::= & \nt{atom}\ |\ \lpneg \nt{atom}\ |\ \nt{atom}_1\ \lpand\ \nt{atom}_2\ |\ \nt{atom}_1\ \lpor\ \nt{atom}_2\\
\nt{atom} & ::= & \tm{true}\ |\ \tm{false}\ |\ p(term_1,\ldots,term_n)\ |\ \nt{term}_1\ \odot\ \nt{term}_2 \\
 & &\text{ for $\odot \in \set{\tm{<},\tm{=<},\tm{=/=}, \tm{=:=}, \tm{>=},\tm{>},\tm{=}}$}\\ 
\nt{term} & ::= & x |\ c^t\ |\ \tm{sqrt(} \nt{term} \tm{)} |\ \nt{term}_1\ \oplus\ \nt{term}_2\\
 & &\text{ for $\oplus \in \set{\tm{+},\tm{-},\tm{*},\tm{/},\tm{\lppow}}$}\\
\\
p & ::= & \text{predicate name}\\
x & ::= & \text{variable name}\\
c^t & ::= & \text{constant of type }t\\
a & ::= & \text{name (attribute)}
\end{eqnarray*}
\end{small}
\caption{Abstract syntax of PrivaLog. Here $\kleenepos{z}$ denotes one or more repetitions of $z$.}\label{fig:syntax}
\end{figure}

\paragraph{Limitations} Similarly to Datalog, there are certain additional constraints set on an input program.

\begin{enumerate}
\item\label{prop:definite} \textit{The program is stratified}, i.e. IDB relations can be ordered in such a way that the body of the rules corresponding to the $i$-th relation may only contain $j$-th relations for $j \leq i$. This is a quite standard constraint on Datalog programs~\cite{SEKI1991107}, which allows to apply certain logic derivation strategies.
\item\label{prop:groundneg} \textit{Safety.} Every variable that appears in the rule head must also appear in the rule body. Every variable that appears in a negated atom must also appear in a positive atom. This is a quite standard constraint on Datalog programs~\cite{Groiss93aformal}.
\item\label{prop:semipos} \textit{Negations are not allowed in front of formulae containing IDB predicates.} This is a more strict constraint than the standard Datalog assumption which requires that the body of the rules corresponding to the $i$-th relation may only contain negations of $j$-th relations for $j < i$. In particular, we only allow $j = 0$, i.e. a negated relation can only be an EDB relation. 
\end{enumerate}
\subsection{Semantics}

Intuitively, semantics of a logic program defines the set of sentences that are true/false w.r.t. the program. Since we will be transforming a declarative logic program to an imperative program, we need to be able to compare them semantically. As the basis, we take Datalog relational semantics, described e.g. in~\cite{semantics}. This seems well suitable for comparing a Datalog and a SecreC program, since in our case the generated SecreC program will also output a certain relation.



In Datalog, a fact is considered false if it does not follow from the given set of fact and rules. This approach is called \emph{negation by failure}. In the world of databases, a relation is considered false if it does not belong to the set of true relations.

\begin{definition}[negation~\cite{semantics}] For any schema $\Gamma$, there exists a universal relation $\univ_{\Gamma}$. Negation
of a relation $R$ can then be defined as $\neg R := \univ_{\Gamma} \setminus R$. 
\end{definition}

Let $\rel_{\Gamma}$ be the set of all subsets of $\univ_{\Gamma}$. A Datalog formula can be treated as a query operating on $\rel_{\Gamma}$.
\begin{definition}[formula semantics~\cite{semantics}]\label{def:formulasem} A formula $q$ whose free term variables are contained in $\Gamma$ denotes a function from $\rel^n_{\Gamma}$ to $\rel_{\Gamma}$.
\[\sempar{\cdot}_{\Gamma} : \nt{formula} \to \rel^n_{\Gamma} \to \rel_{\Gamma}\enspace.\]
If $\R = (\R_1,\ldots,\R_n)$ is a choice of a relation $\R_i$ for each of the literals $A_i$, $\sempar{q}_{\Gamma}(\R)$ is inductively defined according to the rules in Figure~\ref{fig:formsemantics}.
\end{definition}
\begin{figure}
\begin{minipage}[c]{0.45\textwidth}
\begin{eqnarray*}
\sempar{\top}_{\Gamma}(\R) & := & \univ_{\Gamma}\\
\sempar{\bot}_{\Gamma}(\R) & := & \emptyset_{\Gamma}\\
\sempar{A_j}_{\Gamma}(\R) & := & \R_j\\
\end{eqnarray*}
\end{minipage}
\hfil
\begin{minipage}[c]{0.45\textwidth}
\begin{eqnarray*}
\sempar{q_1 \wedge q_2}_{\Gamma}(\R) & := & \sempar{q_1}_{\Gamma}(\R) \cap \sempar{q_2}_{\Gamma}(\R)\\
\sempar{q_1 \vee q_2}_{\Gamma}(\R) & := & \sempar{q_1}_{\Gamma}(\R) \cup \sempar{q_2}_{\Gamma}(\R)\\
\sempar{\neg q}_{\Gamma}(\R) & := & \univ_{\Gamma} \setminus \sempar{q}_{\Gamma}(\R)
\end{eqnarray*}
\end{minipage}

\[\sempar{\exists x.q}_{\Gamma}(\R) := \pi_{\Gamma}(\sempar{q}_{\Gamma \cup \set{x}}(\R))\]
\caption{Semantics of a formula}\label{fig:formsemantics}
\end{figure}

A Datalog rule can be viewed as a formula that generates new relations from already existing relations. A common way (albeit not the only one) to execute a Datalog program is to start from the initial set of facts, and applying the rules to the facts generated so far to obtain more facts until the goal is achieved (bottom-up approach). So-called \emph{immediate consequence operator} computes one such step.

\begin{definition}[immediate consequence operator~\cite{semantics}]\label{def:ico}
Given a program $\mathbb{P} = (P_1,\ldots,P_n)$, where $P_i$ is a predicate with schema $\Gamma_i$, the immediate consequence operator $\ico : \rel^n \to \rel^n$ is defined as
\[\ico(\R_1,\ldots, \R_n) = (\sempar{P_1}_{\Gamma_1} (\R_1,\ldots, \R_n),\ldots, \sempar{P_n}_{\Gamma_n} (\R_1,\ldots, \R_n))\enspace.\]
\end{definition}

Applying the operator repeatedly until no more new facts are generated, we get the program semantics. Let $\mathsf{lfp}_{S}$ denote the least fixpoint operator on set $S$.
\begin{definition}[program semantics]\label{def:progsemantics}
The semantics of a program $\mathbb{P}$ is defined to be
\[\sempar{\mathbb{P}} := \mathsf{lfp}_{\rel^n}(\ico)\]
and may be calculated by iterative application of $\ico$ to $\bot$ until fixpoint is reached.
\end{definition}

Definition~\ref{def:progsemantics} covers all relations that would be true w.r.t. a program $\mathbb{P}$. However, if we have fixed a particular goal, then it is sufficient to consider only those relations that correspond to a certain goal predicate, i.e. one component of $\mathsf{lfp}_{\rel^n}(\ico)$. Let $\pi_{\Gamma}$ denote projection to a schema $\Gamma$.

\begin{definition}[Datalog semantics w.r.t. a predicate]\label{def:progsemanticsgoal}
Let $\mathbb{P} = (P_0,\ldots,P_n)$ be a Datalog program. The semantics of $\mathbb{P}$  w.r.t. predicate $P_j$ is defined to be
\[\sempar{\mathbb{P}} := \pi_{\Gamma_j}(\mathsf{lfp}_{\rel^n}(\ico)[j])\enspace,\]
where $\Gamma_j$ is the schema of $P_j$.
\end{definition}

Existence of the fixpoint is typically ensured by constraining the program such that $\ico$ is monotone. In a privacy-preserving application, the step on which the fixpoint is reached may depend on private data, so ideally it should not be leaked whether a fixpoint has been reached. For that reason, we compute $\sempar{\mathbb{P}}(m) := \ico^m(\bot)$ for some fixed $m \geq 0$. If we underestimate $m$, we can only guarantee $\sempar{\mathbb{P}}(m) \subseteq \sempar{\mathbb{P}}$, so it is possible that we do not find all the answers to our goal query. 

In Figure~\ref{fig:semantics} we show how the definitions above are adapted to a PrivaLog program, defined according to the syntax of Figure~\ref{fig:syntax}. A PrivaLog program $D_1\ldots D_m\ C_1\ldots C_n\ G$ consists of schema declarations $D_i$ followed by rule/fact clauses $C_i$ and a goal $G$. Let us list the key points of the semantics definition:
\begin{itemize}
\item W.l.o.g. we assume that all IDB/EDB predicate names are of the form $p_j$, where $j \in \set{1,\ldots,m+n}$, and there is exactly one clause $C_j$ for each predicate $p_{m+j}$ for $j \in \set{1,\ldots,n}$. This makes the definitions easier to track, and several clauses for the same $p_j$ can anyway be combined into one by introducing disjunctions.
\item A PrivaLog program is executed on particular inputs $v_1,\ldots,v_k$ that valuate arguments $a_1,\ldots,a_k$ of the goal predicate. We introduce a new predicate name $p_0$ which is the goal predicate instantiated with inputs $v_1,\ldots,v_k$. We are interested in the set of satisfiable solutions for $p_0$, which will be the program output.
\item Since table data is not a part of the program syntax, there is an additional input $\tables : \vnames \to \vals^{n_1 \times n_2}$ that contains data. The program is extended with facts $p_j(\tables(p_j)[i,1],\ldots,\tables(p_j)[i,n_2])$, for $i \in \set{1,\ldots,n_1}$, $j \in \set{1,\ldots,m}$.
\item The relation $\R_j$ is the set of satisfiable solutions for $p_j(x_1,\ldots,x_n)$ for $j \in \set{1,\ldots,m+n}$. We use notation $\R_j[i]$ to refer to the $i$-th argument of relation $\R_j$.
\item Semantics of (in)equality predicates is defined as a \emph{filter} over relations. It in turn depends on semantics of arithmetic expressions. 
\end{itemize}


\begin{figure*}
\begin{minipage}[c]{0.45\textwidth}
\begin{eqnarray*}
\sempar{\tm{true}}_{\Gamma}(\R) & := & \univ_{\Gamma}\\
\sempar{\tm{false}}_{\Gamma}(\R) & := & \emptyset_{\Gamma}\\ \\
\end{eqnarray*}
\end{minipage}
\hfil
\begin{minipage}[c]{0.45\textwidth}
\begin{eqnarray*}
\sempar{q_1 \lpand q_2}_{\Gamma}(\R) & := & \sempar{q_1}_{\Gamma}(\R) \cap \sempar{q_2}_{\Gamma}(\R)\\
\sempar{q_1 \lpor q_2}_{\Gamma}(\R) & := & \sempar{q_1}_{\Gamma}(\R) \cup \sempar{q_2}_{\Gamma}(\R)\\
\sempar{\lpneg q}_{\Gamma}(\R) & := & \univ_{\Gamma} \setminus \sempar{q}_{\Gamma}(\R)\\
\end{eqnarray*}
\end{minipage}
\begin{eqnarray*}
\sempar{p_j(t_1,\ldots,t_n)}_{\Gamma}(\R) & := & \pi_{\Gamma}(\sigma_{\sempar{t_1 =:= \R_{j}[1] \lpand \ldots \lpand t_n =:= \R_{j}[n]}}( \R_j \times \univ_\Gamma))\\
\sempar{t_1\ \odot\ t_2}_{\Gamma}(\R) & := & \sigma_{\sempar{t_1\ \odot\ t_2}} \univ_{\Gamma}\text{ for }\odot\in\set{\tm{<},\tm{=<},\tm{=/=}, \tm{=:=}, \tm{>=},\tm{>}}\\
\sempar{t_1\ \tm{=}\ t_2}_{\Gamma}(\R) & := & \sigma_{\sempar{t_1\ \tm{=:=}\ t_2}} \univ_{\Gamma}  
\end{eqnarray*}

\begin{eqnarray*}
\sempar{c^t}\ r & := & c^t\text{ for a constant }c^t\\
\sempar{x}\ r & := & r[x]\text{ for a variable }x\\
\sempar{\lpneg\ t}\ r & := & \neg\sempar{t}\ r\\
\sempar{t_1\ \lpand\ t_2}\ r & := & \sempar{t_1}\ r \wedge \sempar{t_2}\ r\\
\sempar{t_1\ \lpor\ t_2}\ r & := & \sempar{t_1}\ r \vee \sempar{t_2}\ r\\
\sempar{t_1\ \oplus\ t_2}\ r & := & \sempar{t_1}\ r\ \oplus\ \sempar{t_2}\ r\text{ for }\oplus \in \set{\tm{+},\tm{-},\tm{*},\tm{/}, \lppow}\\
\sempar{t_1\ \odot\ t_2}\ r & := & \sempar{t_1}\ r\ \odot\ \sempar{t_2}\ r\text{ for }\odot\in\set{\tm{<},\tm{=<},\tm{=:=},\tm{=/=},\tm{>=},\tm{>}}
\end{eqnarray*}

\begin{eqnarray*}
\sempar{p_j(x_1,\ldots,x_n) \imp\ q} & := & \pi_{\set{x_1,\ldots,x_n}}\sempar{q}_{\vars(q)}
\end{eqnarray*}

\begin{eqnarray*}
\sempar{D_1\,\ldots\,D_m\ C_1\,\ldots\,C_n\ G}\ [v_1,\ldots,v_k]\ \tables & := & \mathsf{lfp}_{\rel^{m+n+1}}(\ico)[0]\\
where & & \\
\ico(\R) & := & (F_0 (\R),\ldots, F_{m+n} (\R))\\
F_0 & := & \sempar{p_0(y_1,\ldots,y_{\ell}) \imp p(v_1,\ldots,v_k,y_1,\ldots,y_\ell)}\text{ for }G=\tm{?-}p(a_1,\ldots,a_k,y_1,\ldots,y_\ell)\\
F_i & := & \sempar{\declared(D_j)\ \tables}\text{ for }j \in \set{1,\ldots,m}\\
F_{m+j} & := & \sempar{C_j}\text{ for }j \in \set{1,\ldots,n}\\
\declared(\imp \tm{type}\ \tm{(} p_j(x_1,\ldots,x_n) \tm{)} \lpend)\ \tables & := & p_j(x_1,\ldots,x_n)\ \imp\ (x_1\ \tm{=}\ \tables(p_j)[1,1]) \lpand \ldots \lpand\ (x_n\ \tm{=}\ \tables(p_j)[1,n])\ \lpor \\
& & \indnt\indnt\indnt\cdots\lpor \\
& & \indnt\indnt\indnt(x_1\ \tm{=}\ \tables(p_j)[\size(p_j),1]) \lpand \ldots \lpand\ (x_n\ \tm{=}\ \tables(p_j)[\size(p_j),n]) \lpend
\end{eqnarray*}
\caption{Relational semantics of a PrivaLog program with table declarations $D_1,\ldots,D_m$, clauses $C_1,\ldots,C_n$, and a goal $G$, executed on inputs $v_1,\ldots,v_k$. Here $\size(p)$ is the number of rows in the table $p$, and $\vars(q)$ is the set of all variables of $q$.}\label{fig:semantics}
\end{figure*}

\paragraph{Handling exceptions.} So far, we have defined semantics only for well-defined programs. Execution of real programs may result in an error if we try to evaluate an expression that contains free variables (e.g. $\tm{X < 5}$ for a free variable \texttt{X}), try to combine non-matching types ($\tm{X \tm{=} 2 \tm{+} a}$), or if some arithmetic function is undefined for given input ($\tm{X = Y / 0}$). While the first two kinds of error depend on program structure and can be detected in the preprocessing phase, it is more difficult with arithmetic errors which depend on private data. We cannot output immediate errors in such case since it would require observing private values. Instead, we just let the resulting SecreC program produce a garbage result, without announcing that it is actually garbage. In general, we require correctness only for valid PrivaLog programs executed on valid inputs.

\section{Preprocessing of a PrivaLog program}\label{sec:trans}

The goal of preprocessing is to rewrite the program in such a way that it would be easily convertible to a SecreC program. There are the certain properties, summarized in Definition~\ref{def:inlined}, that the program should satisfy after preprocessing.

\begin{definition}\label{def:inlined}
A PrivaLog program defined according to the syntax of Fig.~\ref{fig:syntax} is called \emph{preprocessed} if its every rule satisfies the following properties.
\begin{enumerate}
\item\label{it:wellformed:disj} In every disjunction \texttt{X \lpor\ Y} occurring in the rule, both \texttt{X} and \texttt{Y} are ground terms, i.e. do not contain free variables.
\item\label{it:wellformed:asgns} In every unification \texttt{X = Y} occurring in the rule, the term \texttt{X} is a free variable, and \texttt{Y} does not contain any free variables.
\item\label{it:wellformed:unfold} The rule does not contain any IDB predicates.
\end{enumerate}
\end{definition}
These conditions are achieved by applying a sequence of certain transformations, which we now describe one by one.

\subsection{Elimination of non-deterministic choices}\label{sec:nondet}

In logic programs, if a disjunction \texttt{X \lpor\ Y} contains a fresh variable \texttt{Z} that has not used anywhere before, then we are having a choice, where the value of \texttt{Z} is being assigned non-deterministically. We want to avoid such constructions in SecreC. A simple way to resolve this problem is to bring the RHS $q$ of each rule $p \imp q$ to \emph{disjunctive normal form} (DNF), i.e. $\bigvee_i \bigwedge_j A_{ij}$ where $A_{ij}$ are literals of $q$. Putting $q_i = \bigwedge_j A_{ij}$, we split the rule $p \imp q$ to several rules $p \imp q_1 \lpend\ \ldots\ p \imp q_n \lpend$, where each new rule does not contain any disjunctions


We need to be careful that the literals of $q_i$ keep the same order which they have had in $q$, so that each $q_i$ corresponds to a single search path. To keep the initial order of literals, we bring a formula to DNF using basic laws of logic. First, we bring the formula to \emph{negation normal form} using De Morgan's laws to push negation inwards, and eliminating double negations. After that, distributivity laws can be used to push conjunctions inward. All applied rewrite rules do not change the initial order of literals in the formula. The resulting program satisfies property~(\ref{it:wellformed:disj}) of Def.~\ref{def:inlined}.

\subsection{Adornment of rules}\label{sec:adorn}

During the transformation, we will need to know which variables are free and which are bounded. In terms of information flow, bounded variables are treated as inputs, and free variables as outputs. Assigning bounded (\texttt{b}) and free (\texttt{f}) labels to variables and predicate arguments is called \emph{adornment}. We use a slightly modified version of the Basic Sideways Capture Rule (BSCR)~\cite{rulegoal} to add boundness annotations to the variables in rules. The adornment algorithm is given in Figure~\ref{fig:adorn}.

\begin{figure*}
\begin{eqnarray*}
\ell_{\B}(x) & = & \begin{cases}\tm{b}\text{ if }x \in \B\\ \tm{f}\text{ otherwise }\end{cases}\\
\lbl_{\B}(q) & = & q\text{ with every variable $x \in \vars(q) \cap \B$ extended with a label $\ell_{\B}(x)$}\\
\pred(p(\ldots) \imp q\lpend) & = & p
\end{eqnarray*}
\begin{eqnarray*}
\adornf{\mathbb{P}}{p(x_1,\ldots,x_n)}\ \B & = & (p(x_1\ \tm{:}\ \ell_{\B}(x_1),\ldots,x_n\ \tm{:}\ \ell_{\B}(x_n)),\ \B \cup \set{x_1,\ldots,x_n},\ \LL)\\
where & & \LL = \bigcup_{i \in I}\adornc{\mathbb{P}}{C_i, \ell_{\B}(x_1)\ldots\ell_{\B}(x_n)}\\
& & I = \set{i\ |\ C_i \in \mathbb{P}, \pred(C_i) = p}\\
\\
\adornf{\mathbb{P}}{q_1 \lpand q_2}\ \B & = &  (q'_1 \lpand q'_2,\ \B'_1 \cup \B'_2,\ \LL_1 \cup \LL_2)\\
where & & (q'_1,\B'_1,\LL_1) = \adornf{\mathbb{P}}{q_1}\ \B\\
      & & (q'_2,\B'_2,\LL_2) = \adornf{\mathbb{P}}{q_2}\ \B'_1\\
\\
\adornf{\mathbb{P}}{x_1\tm{ = }x_2}\ \B & = & \begin{cases}(x_1\tm{ : b =:= }x_2\tm{ : b},\ \B,\ \emptyset) \text{ if }x_1 \in \B, x_2 \in \B\\
(x_1\tm{ : f = }x_2\tm{ : b},\ \B\cup\set{x_1},\ \emptyset) \text{ if }x_1 \notin \B, x_2 \in \B\\
(x_2\tm{ : f = }x_1\tm{ : b},\ \B\cup\set{x_2},\ \emptyset) \text{ if }x_1 \in \B, x_2 \notin \B\\
\bot\text{ otherwise }\end{cases}\\
\adornf{\mathbb{P}}{x\tm{ = }e}\ \B & = & \begin{cases}(x\tm{ : b =:= }\lbl_{\B}(e),\ \B,\ \emptyset) \text{ if }x \in \B\text{ and }\vars(e) \subseteq \B\\
(x\tm{ : f = }\lbl_{\B}(e),\ \B \cup \set{x},\ \emptyset) \text{ if }x \notin \B\text{ and }\vars(e) \subseteq \B\\ \bot\text{ otherwise }\end{cases}\\
\adornf{\mathbb{P}}{e\tm{ = }x}\ \B & = & \adornf{\mathbb{P}}{x = e}\ \B\\
\adornf{\mathbb{P}}{q}\ \B & = & \begin{cases}(\lbl_{\B}(q),\ \B,\ \emptyset) \text{ if }\vars(q) \subseteq \B\\
\bot\text{ otherwise }\end{cases}\\
& &\text{ for any other }q
\end{eqnarray*}
\begin{eqnarray*}
\adornc{\mathbb{P}}{p(x_1,\ldots,x_n) \imp q,\ \ell_1\ldots\ell_n} & = & \set{p_{\ell_1\ldots\ell_n}(x_1\tm{ : }\ell_1,\ldots,x_n\tm{ : }\ell_n) \imp q'} \cup \LL\\
where & & (q',\_,\LL) = \adornf{\mathbb{P}}{q}\ \set{x_i\ |\ \ell_i = \tm{b}}\\ \\
\adorng{\mathbb{P}}{\tm{?-}p(a_1,\ldots,a_k,x_1,\ldots,x_\ell)} & = & (\tm{?-}p_{\ell_1\ldots\ell_n}(a_1\tm{ : b},\ldots,a_k\tm{ : b},x_1\tm{ : f},\ldots,x_\ell\tm{ : f}),\ \bigcup_{i \in I}\adornc{\mathbb{P}}{C_i, \ell_1\ldots\ell_n})\\
where & & I = \set{i\ |\ C_i \in \mathbb{P}, \pred(C_i) = p}\\
& & \ell_i = \begin{cases}\tm{b}\text{ for }i \leq k \\\tm{f}\text{ otherwise }\end{cases}\\
\\
\adornp{D_1\ldots D_m\ C_1\ \ldots C_n\ G} & = & D_1\ldots D_m\ C'_1\ \ldots C'_{n'}\ G'\\
where & & (G',\set{C'_1,\ldots,C'_{n'}}) = \adorng{\mathbb{P}}{G}
\end{eqnarray*}
\caption{Adornment of a PrivaLog program $\mathbb{P} = D_1\ldots D_m\ C_1\ \ldots C_n\ G$.}\label{fig:adorn}
\end{figure*}

Let us describe the intuition behind adornment. We start from a rule $p \imp q_1 \lpand \ldots \lpand q_k$ that corresponds to the goal predicate (there may be several such rules). W.l.o.g. we assume that all arguments $z_1,\ldots,z_n$ of $p$ are variables (if any of them is an expression $e_i$, we can introduce a fresh variable $z_i$ and add $z_i = e_i$ to the rule body).

We create a new adorned rule $\mathtt{p}_{l_1,\ldots,l_n}$, where $l_i \in \set{\mathtt{b},\mathtt{f}}$ is a label which tells whether the $i$-th argument of the rule head is a free variable (i.e. will be treated as an output) or a bounded variable (i.e. will be treated as an input). The labels $l_i$ are determined by the goal predicate. We then traverse through the rule body $q_1 \lpand \ldots \lpand q_k$. Assuming that disjunctions have been eliminated as described in Sec.~\ref{sec:nondet}, every $q_i$ is either a single EDB/IDB predicate, a comparison, or a negation. Depending on the case, the algorithm acts as follows:
\begin{itemize}
\item If $q_i$ is a comparison or a negation, then we require that all its variables are bounded (since e.g. we cannot evaluate $X > 5$ if $X$ is a free variable, and Datalog would raise runtime exception in such a case). An exception is the unification $X = Y$, which is also treated as an assignment. If $X$ is a free variable, it gets the label of $Y$. If $Y$ is a free variable, it gets the label of $X$. We assume that $X$ and $Y$ cannot be free variables both at once since then $Y$ could as well be renamed to $X$ in the entire clause.
\item If $q_i$ is an EDB predicate, then all variables of $q_i$ that have been free so far become bounded, since they get their values from the database.
\item If $q_i$ is an IDB predicate, then we in addition repeat the procedure recursively with all rules that match the head of $q_i$. It is possible that the corresponding rules have already been adorned before, and we only need to repeat the process if we call these rules with a different adornment pattern.
\end{itemize}
It is possible that we get different adornments for the same IDB predicate. E.g. a rule \texttt{p(X,Y) :- X = Y} may be transformed to three rules \texttt{p\_bf(X,Y)}, \texttt{p\_fb(X,Y)}, \texttt{p\_bb(X,Y)}, where the first one assigns $Y := X$, the second one $X := Y$, and the last one performs a comparison $X == Y$.

Let us demonstrate adornment on our running example. For shortness, in the following, \texttt{*} denotes a bounded variable. Note that the same variable can be free in one part of the rule and bounded in the other part. For example, the variable \texttt{Ship} in the rule \texttt{arrival\_fbbf} is initially free, but it gets valuated after the EDB predicate \texttt{ship} is consulted.

\begin{small}
\begin{verbatim}
reachability_time_bbf(Ship*,Port*,Time) :-
    ship(Ship*,X1,Y1,Speed,_,_),
    port(Port*,X2,Y2,_),
    Time = sqrt((X1* - X2*)^2 + (Y1* - Y2*)^2) / Speed*.

suitable_port_bb(Ship*,Port*) :-
    port(Port*,_,_,Capacity),
    ship(Ship*,_,_,_,_,CargoAmount),
    Capacity* >= CargoAmount*.

arrival_fbbf(Ship,Port*,CargoType*,Time) :-
    ship(Ship,_,_,_,CargoType*,_),
    suitable_port(Ship*,Port*),
    reachability_time(Ship*,Port*,Time).

?-arrival_fbbf(Ship,portname* : private string,
               cargotype* : private string,Time).
\end{verbatim}
\end{small}

\subsection{Elimination of IDB predicates}\label{sec:unfold}

While internal function calls and recursion are not a technical limitation of SecreC language, it performs much better if SIMD operations are used, as discussed in Sec.~\ref{sec:secrec}. For that reason, we unfold the IDB predicates as much as possible to achieve parallel evaluation of rules on all possible non-deterministic choices of EDB inputs. Eliminating IDB predicates will also simplify the subsequent steps. The unfolding transformation is close to the one described in~\cite{SEKI1991107}.

The particular transforming algorithm that we are using is based on a naive bottom-up execution of a Datalog program, i.e. the rules are providing the knowledge base with more and more facts until a fixpoint (or the bound on the number of iterations) is reached. Differently from standard bottom-up evaluation, we treat private data as symbolic variables, which means that we cannot evaluate all literals (we can only do partial evaluation like constant propagation and constant folding), and in general we  still get \emph{rules} instead of \emph{facts}. However, these rules do not contain IDB predicates anymore. In the end, we leave only those rules that are needed by the goal (there may be several such rules, but all of them have the same head predicate). After this procedure, the program satisfies property~(\ref{it:wellformed:unfold}) of Def.~\ref{def:inlined}.

Let us demonstrate unfolding on our running example. We assume that adornments have been computed as in Sec.~\ref{sec:adorn}, and \texttt{*} denotes a bounded variable. We note that the same rule can be expressed by a least number of calls of EDB predicates \texttt{ship} and \texttt{port}, but that rewriting is delegated to a separate optimization which is not a part of formal transformation.

\begin{small}
\begin{verbatim}
arrival_fbbf(Ship,Port*,CargoType*,Time) :-
    ship(Ship,_,_,_,CargoType*,_),
    port(Port*,_,_,Capacity),
    ship(Ship*,_,_,_,_,CargoAmount),
    Capacity* >= CargoAmount*,
    ship(Ship*,X1,Y1,Speed,_,_),
    port(Port*,X2,Y2,_),
    Time = sqrt((X1* - X2*)^2 + (Y1* - Y2*)^2) / Speed*.

?-arrival_fbbf(Ship, portname* : private string,
               cargotype* : private string, Time).
\end{verbatim}
\end{small}

 The declarative definition of the transformation is given in Figure~\ref{fig:trans}. We use notation $\RB \in \clauseset$ for the \emph{rule base}, i.e. the set of ground rules (without IDB predicates) that have already been generated. Formally, we define transformation functions $\transfunp : program \to program$, $\transfunc : clause \to \clauseset$, and $\transfunf : formula \to Set(formula)$.

\begin{figure*}
\begin{eqnarray*}
\transf{\tm{true}}(\RB) & := & \set{\tm{true}}\\
\transf{\tm{false}}(\RB) & := & \emptyset\\
\transf{p(x_1,\ldots,x_n)}(\RB) & := & \set{p(x_1,\ldots,x_n)}\text{ for an EDB predicate $p$}\\
\transf{p_j(x_1,\ldots,x_n)}(\RB) & := & \{q\theta\ \lpand\ (x_1 = y_{i1}\theta)\ \lpand \ldots \lpand\ (x_n = y_{in}\theta)\ |\ p_j(y_{i1},\ldots,y_{in})\ \imp\ q \in \RB_j\}\text{ for an IDB predicate $p_j$}\\
& & \text{where }\theta\text{ is a substitution that gives fresh names to all variables}\\
\transf{t_1\ \odot\ t_2}(\RB) & := & \set{t_1 \odot\ t_2}\text{ for }\odot\in\set{\tm{<},\tm{=<},\tm{=:=},\tm{=/=},\tm{>=},\tm{>},\tm{=}}\\ 
\transf{q_1\ \lpand\ q_2}(\RB) & := & \set{q'_1\ \lpand\ q'_2\ |\ q'_1 \in \transf{q_1}(\RB), q'_2 \in \transf{q_2}(\RB)}\\
\transf{q_1\ \lpor\ q_2}(\RB) & := & \set{q_1\ \lpor\ q_2}\\
\transf{\lpneg q}(\RB) & := & \set{\lpneg q}
\end{eqnarray*}
\begin{eqnarray*}
\transc{p_j(x_1,\ldots,x_n)\ \imp\ q\ \tm.}(\RB) & := & \set{p_j(x_1,\ldots,x_n)\ \imp\ q'\ \lpend\ |\ q' \in \transf{q}(\RB)}
\end{eqnarray*}
\begin{eqnarray*}
\transp{D_1\,\ldots\,D_m\ C_1\,\ldots\,C_n\ G} & := & D_1 \ldots D_m\ \mathsf{lfp}_{\clauseset^n}(\sico)\ G\\
where\hfill & & \\
\sico(\RB) & = & (F_1(\RB), \ldots, F_{n}(\RB))\\
F_i & := & \transc{C_i}\text{ for }i \in \set{1,\ldots,n}\\
\end{eqnarray*}

\caption{Unfolding a PrivaLog program.}\label{fig:trans}
\end{figure*}

The nature of transformation of Figure~\ref{fig:trans} is very similar to the semantics described in Figure~\ref{fig:semantics}. The main difference is that $\transfunp$ does not utilize the actual table data, but postpones valuation of table relations. Intuitively, on each step of the iteration of $\ico$ and $\sico$, the sets $\R$ and $\RB$ are tightly related to each other, so that we would always be able to get $\R$ out of $\RB$ as soon as we get the data.

\begin{definition}[relations derived w.r.t. data]\label{def:derivedwrt}
Let $\R$ be a set of relations over a predicate $p$, and let $\RB$ be a set of rules for $p$. Let $\tables$ be a set of extensional database relations. We say that $\R$ is derived from $\RB$ w.r.t. $\tables$, denoting $\R \derwrt{\tables} \RB$, iff
\begin{eqnarray*}
\R = \{(a_{i1},\ldots,a_{in}) & | & p(x_{i1},\ldots,x_{in})\ \imp\ q_i\ \in \RB,\\
& & (a_{i1},\ldots,a_{in}) \in \sempar{q_i}_{\set{x_{i1},\ldots,x_{in}}}(\tables)\}\enspace.
\end{eqnarray*}
\end{definition}

We want to show that the property of Definition~\ref{def:derivedwrt} is maintained after applying the transformation. First, we show that it holds on formula level.

\begin{lemma}\label{lm:transf}
Let $\R \derwrt{\tables} \RB$, $\tables \subseteq \R$. Let $\Gamma = \vars(q)$. Then,
\[\sempar{q}_{\Gamma}(\R) = \bigcup_{q' \in \transf{q}(\RB)} \sempar{q'}_{\Gamma}(\tables)\enspace.\]
\end{lemma}
\begin{proof}
We prove the statement by induction on size of $q$. We will refer to the left hand side of the proved equality as $LHS$, and the right hand side as $RHS$, so we need to show that $LHS = RHS$.
\begin{itemize}
\item Let $q \in \set{\mathsf{true}, \mathsf{false}, \lpneg q', q_1\ \lpor\ q_2, \mathsf{t_1 \odot t_2}\text{ for }\odot\in\set{\tm{<},\tm{=<},\tm{=:=},\tm{>=},\tm{>},\tm{=}}}$. In these cases, the transformation does not change $q$, and we have $\sempar{q}_{\Gamma}(\tables) = \sempar{q}_{\Gamma}(\R)$, and $\transf{q}(\RB) = \set{q}$. Hence, 
\[RHS = \bigcup_{q' \in \transf{q}(\RB)} \sempar{q'}_{\Gamma}(\tables) = \sempar{q}_{\Gamma}(\tables) = \sempar{q}_{\Gamma}(\R) = LHS\enspace.\]

\item Let $q = p_j(x_1,\ldots,x_n)$. We have $\Gamma = Var(q) = (x_1,\ldots,x_n)$ By definition,
\[\transf{p_j(x_1,\ldots,x_n)}(\RB) = \set{q_i\theta\ \lpand\ (x_1 = y_{i1}\theta)\ \lpand \ldots \lpand\ (x_n = y_{in}\theta)\ |\ p_j(y_{i1},\ldots,y_{in})\ \imp\ q_i \in \RB}\enspace,\]
where $\theta$ issues fresh names to all variables. We need to show
\begin{eqnarray*}
\sempar{q}_{x_1,\ldots,x_n}(\R) & = & \bigcup_{p_j(y_{i1},\ldots,y_{in})\ \imp\ q_i \in \RB}\sempar{q_i\theta\ \lpand\ (x_1 = y_{i1}\theta)\ \lpand \ldots \lpand\ (x_n = y_{in}\theta)}_{x_1,\ldots,x_n}(\tables)\enspace.
\end{eqnarray*}
By assumption of the lemma, $\R \derwrt{\tables} \RB$, so we have
\[\R_j = \set{(a_{i1},\ldots,a_{in})\ |\ p_j(y_{i1},\ldots,y_{in})\ \imp\ q_i\ \in \RB,\ (a_{i1},\ldots,a_{in}) \in \sempar{q_i}_{\set{y_{i1},\ldots,y_{in}}}(\tables)}\enspace.\]
We have
\begin{eqnarray*}
LHS & = & \sempar{q}_{x_1,\ldots,x_n}(\R) = \sempar{p_j(x_1,\ldots,x_n)}_{x_1,\ldots,x_n}(\R) = \R_j\\
& = & \set{(a_{i1},\ldots,a_{in})\ |\ p_j(y_{i1},\ldots,y_{in})\ \imp\ q_i\ \in \RB,\ (a_{i1},\ldots,a_{in}) \in \sempar{q_i}_{\set{y_{i1},\ldots,y_{in}}}(\tables)}\\
& = & \bigcup_{p_j(y_{i1},\ldots,y_{in})\ \imp\ q_i\ \in \RB} \sempar{q_i}_{\set{y_{i1},\ldots,y_{in}}}(\tables)\enspace;
\end{eqnarray*}
\begin{eqnarray*}
RHS & = & \bigcup_{p_j(y_{i1},\ldots,y_{in})\ \imp\ q_i \in \RB}\sempar{q_i \theta\ \lpand\ (x_1 = y_{i1}\theta)\ \lpand \ldots \lpand\ (x_n = y_{in}\theta)}_{x_1,\ldots,x_n}(\tables)\\
& = & \bigcup_{p_j(y_{i1},\ldots,y_{in})\ \imp\ q_i \in \RB}\sempar{q_i \theta\ \lpand\ (x_1 = y_{i1}\theta)\ \lpand \ldots \lpand\ (x_n = y_{in}\theta)}_{y_{i1}\theta,\ldots,y_{in}\theta}(\tables)\\
& = & \bigcup_{p_j(y_{i1},\ldots,y_{in})\ \imp\ q_i \in \RB}\sempar{q_i \theta}_{y_{i1}\theta,\ldots,y_{in}\theta}(\tables)\enspace.
\end{eqnarray*}
In the last equality, since $x_i$ are not present in $q_i$, the conditions $x_j = y_{ij}$ are always true and do not affect the set of relations generated by $q_i$, so they only affect $\Gamma$. Since $\theta$ only renames variables, without changing the semantics, we have LHS = RHS.

\item Let $q = q_1\ \lpand q_2$. By induction hypothesis,
    \begin{itemize}
    \item $\sempar{q_1}_{\Gamma}(\R) = \bigcup_{q' \in \transf{q_1}(\RB)} \sempar{q'}_{\Gamma}(\tables)$;
    \item $\sempar{q_2}_{\Gamma}(\R) = \bigcup_{q' \in \transf{q_2}(\RB)} \sempar{q'}_{\Gamma}(\tables)$.
    \end{itemize}
    Hence,
    \begin{eqnarray*}
          LHS & = & \sempar{q_1\ \lpand q_2}_{\Gamma}(\R) = \sempar{q_1}_{\Gamma}(\R) \cap \sempar{q_2}_{\Gamma}(\R)\enspace;\\
          RHS & = & \bigcup_{q' \in \transf{q_1 \lpand q_2}(\RB)} \sempar{q'}_{\Gamma}(\tables)\\
              & = & \bigcup_{q_1 \in \transf{q_1}(\RB), q_2 \in \transf{q_2}(\RB)} \sempar{q_1}_{\Gamma}(\tables) \cap \sempar{q_2}_{\Gamma}(\tables)\\
              & = & \bigcup_{q_1 \in \transf{q_1}(\RB)} \sempar{q_1}_{\Gamma}(\tables) \cap \bigcup_{q_2 \in \transf{q_2}(\RB)} \sempar{q_2}_{\Gamma}(\tables)\\
              & = & \sempar{q_1}_{\Gamma}(\R) \cap \sempar{q_2}_{\Gamma}(\R)\enspace.\qedhere
    \end{eqnarray*}


\end{itemize}
\end{proof}

We now show that Definition~\ref{def:derivedwrt} is maintained after applying the transformation on clause level.
\begin{lemma}\label{lm:transc}
Let $\R \derwrt{\tables} \RB$, $\tables \subseteq \R$. Then, $\sempar{q}_{\set{x_1,\ldots,x_n}}(\R) \derwrt{\tables} \transc{p(x_1,\ldots,x_n)\ \imp\ q\ \lpend}(\RB)$.
\end{lemma}
\begin{proof}
We need to show that
\[LHS := \sempar{q}_{\set{x_1,\ldots,x_n}}(\R) = \bigcup_{p(x_{i1},\ldots,x_{in})\ \imp\ q_i \in \transc{p(x_1,\ldots,x_n)\ \imp\ q}(\RB)} \sempar{q_i}_{\set{x_{i1},\ldots,x_{in}}}(\tables) =: RHS\enspace.\]

\begin{itemize}
\item By Lemma~\ref{lm:transf}, since $\R \derwrt{\tables} \RB$, we have $LHS = \sempar{q}_{\set{x_1,\ldots,x_n}}(\R) = \bigcup_{q_i \in \transf{q}(\RB)}\sempar{q_i}_{\set{x_1,\ldots,x_n}}(\tables)$.
\item By definition of clause transformation,
\[\transc{p_j(x_1,\ldots,x_n)\ \imp\ q\ \tm.}(\RB) := \set{p_j(x_1,\ldots,x_n)\ \imp\ q_i\ \lpend\ |\ q_i \in \transf{q}(\RB)}\enspace.\]
Hence,
\[RHS = \bigcup_{q_i \in \transf{q}(\RB)} \sempar{q_i}_{\set{x_{i1},\ldots,x_{in}}}(\tables)\enspace.\]
\end{itemize}
Since $x_j$ and $x_{ij}$ for all $i,j$ correspond to the same relation attribute, we get $LHS=RHS$.
\end{proof}

We now show that the unfolding transformation preserves program semantics. The idea is that, since the property $\R \derwrt{\tables} \RB$ is maintained for each clause on each transformation step, we can show that $\R \derwrt{\tables} \RB$ still holds after the $n$-th steps, which means that the set of relations generated by the initial and the transformed PrivaLog programs are the same. We note that we do not explicitly the assumption that a program is stratified, since according to our semantics we would not obtain new relations from looping constructions for the initial program as well.
\begin{theorem}
Let $\mathbb{P}$ be a program defined according to the syntax of Figure~\ref{fig:syntax}. We have $\sempar{\mathbb{P}} = \sempar{\transp{\mathbb{P}}}$.
\end{theorem}
\begin{proof}
Let $\mathbb{P} = (D_1\ \ldots\ D_m\ C_1\ \ldots\ C_n\ G)$. The idea is to show that $\R \derwrt{\tables} \RB$ is preserved on each step of the $\mathsf{lfp}$ computation.
\begin{itemize}
\item On the first step, the fixpoint operator is applied to $\bot$. We trivially have $\bot \derwrt{\tables} \bot$.
\item Let $\R \derwrt{\tables} \RB$. Let $P_i = p_i(x_{i1},\ldots,x_{ik_i})\ \imp q_i$ be the $i$-th rule. We have:
\begin{eqnarray*}
\ico(\R) & = & (\sempar{q_1}_{\set{x_{11},\ldots,x_{1k_1}}}(\R),\ldots,\sempar{q_n}_{\set{x_{n1},\ldots,x_{nk_n}}}(\R));\\
\sico(\RB) & = & (\transc{p_1(x_{11},\ldots,x_{1k_1})\ \imp q_i}(\RB), \ldots, \transc{p_n(x_{n1},\ldots,x_{nk_n})\ \imp q_n}(\RB)).\\
\end{eqnarray*}
Since $\R \derwrt{\tables} \RB$, by Lemma~\ref{lm:transc}, we have $\sempar{q_i}_{\set{x_{i1},\ldots,x_{ik_i}}}(\R) \derwrt{\tables} \transc{p_i(x_{i1},\ldots,x_{in})\ \imp\ q_i\ \lpend}(\RB)$ for all $i$.
\end{itemize}
It remains to show that the property $\R \derwrt{\tables} \RB$ gives us what is needed. Let $n$ be any iteration step. Consider the set of rules $\RB^n$ obtained for $\transp{\mathbb{P}}$ after the $n$-th step of transformation. We have shown above that $\R^n \derwrt{\tables} \RB^n$. By definition of $\derwrt{\tables}$, we have 
\[\R^n = \set{(a_{i1},\ldots,a_{in})\ |\ p(x_{i1},\ldots,x_{in})\ \imp\ q_i\ \in \RB^n,\ (a_{i1},\ldots,a_{in}) \in \sempar{q_i}_{\set{x_{i1},\ldots,x_{in}}}(\tables)}\enspace.\]

\begin{itemize}
\item Let us prove that $\sempar{\mathbb{P}} \subseteq \sempar{\transp{\mathbb{P}}}$. Let $r \in \R^n = \ico_{\mathbb{P}}^n(\bot)$ be any relation considered to be true after the $n$-th step in the initial program $\mathbb{P}$. We show that $r \in \ico_{\transp{\mathbb{P}}}(\tables)$. Since $\R^n \derwrt{\tables} \RB^n$, there exists a rule $p(x_{i1},\ldots,x_{in})\ \imp\ q_i \in \RB^n$, such that $r = (a_{i1},\ldots,a_{in}) \in \sempar{q_i}_{\set{x_{i1},\ldots,x_{in}}}(\tables)$. Hence, the relation $r$ can be obtained from rules $\RB^n$ that have been generated on the $n$-th transformation step

\item Let us prove that $\sempar{\transp{\mathbb{P}}} \subseteq \sempar{\mathbb{P}}$. Let $r$ be any relation considered to be true in the transformed program $\transp{\mathbb{P}}$. Assume that $r$ has been obtained using any rule from the set $\RB^n$ for some $n$. By construction, the rules of $\RB^n$ do not contain any recursive calls, so their truthness depends only on the extensional database $\tables$. Therefore, $r = (a_1,\ldots,a_n)$ where $p(x_{i1},\ldots,x_{in})\ \imp\ q_i\ \in \RB^n$, and $(a_{i1},\ldots,a_{in}) \in \sempar{q_i}_{\set{x_{i1},\ldots,x_{in}}}(\tables)$. Since $\R^n \derwrt{\tables} \RB^n$, by definition of $\derwrt{\tables}$, we have $r \in \R^n$.
\end{itemize}
\end{proof}

The efficiency of transformation of Figure~\ref{fig:trans} highly depends on the exact order in which the rules are unfolded. Suppose that we have a rule $r \imp q$ which contains IDB predicates $p_1,\ldots,p_k$ as a part of $q$. In our translator, we are following three possible strategies:
\begin{itemize}
\item \textbf{Depth-first-search} (\texttt{dfs}). On one iteration, we only unfold the first IDB predicate $p_1$. If the first predicate is recursive, the algorithm will never converge on a fixpoint.

\item \textbf{Breadth-first-search} (\texttt{bfs}). On one iteration, we unfold each IDB predicate $p_1,\ldots,p_k$ once. This approach requires more memory, and if $p_i$ appears in the head of $j_i\in\mathbb{N}$ different rules, we may end up with $\prod_{i=0}^k j_i$ new rules.

\item \textbf{Ground-first search} (\texttt{gr}). On one iteration, we only unfold the first IDB predicate $p_1$. However, we only allow to substitute $p_1$ with a formula that does not contain any IDB predicates. This approach is the closest to the initial idea of bottom-up execution.

\end{itemize}

Currently, we let the user to choose the strategy. Ideally, we could execute all three in parallel and wait until at least one of them converges.

After elimination of disjunctions, adornment, and unfolding, the program satisfies Definition~\ref{def:inlined}. We can treat the language of the resulting program as a subset of PrivaLog, which we call $\mu$-PrivaLog. The syntax of $\mu$-PrivaLog is given in Fig.~\ref{fig:muplsyntax}. Since unfolding procure ensures that there is exactly one IDB predicate left in the program, we can w.l.o.g. rename that only predicate to \goal. The new construction $bterm$ is a term that does not contain free variables. The semantics of $\mu$-PrivaLog is the same as of semantics PrivaLog, but due to its constrains, it can be simplified. The main difference is that suffices to apply the immediate consequence operator $\ico$ just three times. The first step generates EDB relations. The second step uses these relations to generate all tuples for the only IDB relation \goal. The third step evaluates the goal. In Figure~\ref{fig:muplsemantics}, we define semantics of a $\mu$-PrivaLog program by putting these three steps together.

\begin{figure}
\begin{small}
\begin{eqnarray*}
\nt{program} & ::= & \kleenepos{\nt{schema}}\ \kleenepos{\nt{clause}}\ \nt{goal}\\
\nt{schema} & ::= & \imp\ \tm{type}(p(a_1 : d_1\ t_1,\ldots,a_n\ :\ d_n\ t_n)) \lpend\\
\nt{d} & ::= &  \tm{public}\ |\ \tm{private}\\
\nt{t} & ::= &  \tm{int}\ |\ \tm{bool}\ |\ \tm{float}\ | \tm{string}\\
\nt{goal} & ::= & \tm{?-}\ \goal((a_1\ :\ \tm{b})\ :\ d_1\ t_1,\ldots,(a_n\ : \tm{b}) :\ d_n\ t_n,\\
& & \indnt\indnt x_1\ :\ \tm{f},\ldots,x_m\ :\ \tm{f})\ \lpend\\
\nt{clause} & ::= & \nt{fact}\ |\ \nt{rule}\\
\\
\nt{fact} & ::= &  \goal(c^{t_1}_1,\ldots,c^{t_n}_n) \lpend\\
\nt{rule} & ::= &  \goal(x_1,\ldots,x_n) \imp\ \nt{formula}\ \lpend\\
\nt{formula} & ::= & \nt{atom}\ |\ \lpneg \nt{atom}\ |\ \nt{atom}_1\ \lpand\ \nt{atom}_2\\ 
\nt{atom} & ::= & \tm{true}\ |\ \tm{false}\ |\ p(term_1,\ldots,term_n)\ |\ x\ :\ \tm{f}\ \tm{=}\ bterm\\
& & |\ \nt{bterm}_1\ \odot\ \nt{bterm}_2 \\
 & &\text{ for $\odot \in \set{\tm{<},\tm{=<},\tm{=/=}, \tm{=:=}, \tm{>=},\tm{>}}$}\\ 
\nt{bterm} & ::= & x\ : \tm{b} |\ c^t\ |\ \tm{sqrt(} \nt{bterm} \tm{)} |\ \nt{bterm}_1\ \oplus\ \nt{bterm}_2\\
 & &\text{ for $\oplus \in \set{\tm{+},\tm{-},\tm{*},\tm{/},\tm{\lppow}}$}\\
\nt{term} & ::= & x\ :\ \tm{f}\ |\ bterm
\\
p & ::= & \text{predicate name, }p\neq\goal\\
x & ::= & \text{variable name}\\
c^t & ::= & \text{constant of type }t\\
a & ::= & \text{name (attribute)}
\end{eqnarray*}
\end{small}
\caption{Abstract syntax of $\mu$-PrivaLog. Here $\kleenepos{z}$ denotes one or more repetitions of $z$.}\label{fig:muplsyntax}
\end{figure}

\begin{figure*}
\begin{eqnarray*}
\sempar{D_1\,\ldots\,D_m\ C_1\,\ldots\,C_n\ G}\ [v_1,\ldots,v_k]\ \tables & := & \sempar{F}(\R)[0]\\
where & & \\
\R   & := & (\bot,\R_1,\ldots,\R_m)\\ 
\R_k & := & \set{p_k(\tables(p_k)[i,1],\ldots,\tables(p_k)[i,\arity(p_k)])\ |\ i \in \size(p_k)}\text{ for }k \in \set{1,\ldots,m}\\
F    & := & \bigcup_{j \in \set{1,\ldots,n}}\sempar{C_j\set{x_{j1} \mapsto v_1,\ldots,x_{jk} \mapsto v_k}}\text{ for }C_j = p_j(x_{j1},\ldots,x_{jn}) \imp q_j\lpend
\end{eqnarray*}
\caption{Relational semantics of a $\mu$-PrivaLog program with table declarations $D_1,\ldots,D_m$, clauses $C_1,\ldots,C_n$, and a goal $G$ executed on inputs $v_1,\ldots,v_k$ and a database $\tables$. The semantics of a clause is the same as in Figure~\ref{fig:semantics}, $\size(p)$ is the number of relations in the table $\tables(p)$, and $\arity(p)$ is the number of arguments in the predicate $p$.} \label{fig:muplsemantics}
\end{figure*}

\section{Type system and type inference}\label{sec:typeinf}

Since SecreC is a statically typed language, it is important to ensure the consistency of data types in the PrivaLog program and to warn the user of type errors early. For this, the translator implements a type system with type inference. We cannot introduce intermediate type annotations directly into a PrivaLog program, since the same rule may be called in different contexts, resulting in very different information flows. However, we can do it for a $\mu$-PrivaLog program.

Each expression in a $\mu$-PrivaLog program is assigned a privacy domain and a data type. In addition, each rule head has a return type of \texttt{bool} and a privacy domain. The privacy domain can be either \texttt{public} or \texttt{private}. All possible types are shown on Figure \ref{fig:type-lattice}.

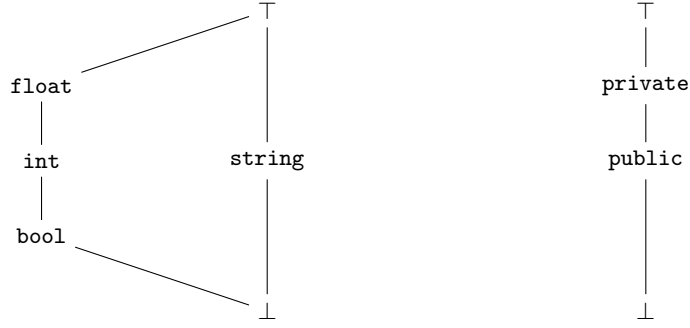
\begin{figure}
  \centering
  \begin{tikzpicture}
    \node (top) at (0,0) {$\top$};
    \node (float32) at (-3,-1) {$\texttt{float}$};
    \node (string) at (0,-2) {$\texttt{string}$};
    \node (int32) at (-3,-2) {$\texttt{int}$};
    \node (bool) at (-3,-3) {$\texttt{bool}$};
    \node (bottom) at (0,-4) {$\bot$};
    \draw (top) -- (float32) -- (int32) -- (bool) -- (bottom);
    \draw (top) -- (string) -- (bottom);
  \end{tikzpicture}
  \hfil
  \begin{tikzpicture}
    \node (top) at (0,0) {$\top$};
    \node (private) at (0,-1) {$\texttt{private}$};
    \node (public) at (0,-2) {$\texttt{public}$};
    \node (bottom) at (0,-4) {$\bot$};
    \draw (top) -- (private) -- (public) -- (bottom);
  \end{tikzpicture}
  \caption{$\mu$-PrivaLog domain and type lattices}
  \label{fig:type-lattice}
\end{figure}

%

The rules for type inference are given in Figure~\ref{fig:typeinfrules}. Let us briefly discuss what these rules are doing.

The rule $DB$ extracts database schema, which is defined by the user, into global context $\Gamma$.

The rule $Goal$ assigns types to input variables of the $\mu$-PrivaLog program as initially specified by the programmer. The rule $Inputs$ propagates these input types to clauses $C_1,\ldots,C_m$ of the program, and each clause $C_j$ stores types in a distinct local context $\Gamma_j$ since variable names may repeat in different clauses.

For clause $C_j$, the types are propagated w.r.t. local context $\Gamma_j$. The rule $Lit$ assigns public domain and type $t$ to a constant of type $t$, and this type is determined by the syntax. The rule $Var$ assigns to bounded variables the domain and type that have been inferred for them before (e.g. by the rule $Asgn$).  The rules $UnOp$, $BinOp$, $Cmp$ are defined in a standard way, and the result of an operator application is assigned private type iff at least one of its arguments has private type. The rule $Asgn$ not only determines the type of the assigned variable, but also the type of assignment itself (as a boolean value), which always evaluates to true and hence has type \tm{public bool}. The most interesting is the rule $Pred$. In a $\mu$-PrivaLog program, all predicates of the form $p(e_1,\ldots,e_n)$ are EDB predicates. Each EDB call may read some data from the database (for a free argument), or make a comparison (for a bounded argument). Assignment of a value to a free variable formally always returns \texttt{true} and hence is a public constant. A comparison inherits its type from the argument and the corresponding EDB attribute.

Finally, the rule $Outputs$ collects the types of outputs from all clauses $C_1,\ldots,C_m$, and determines the types of program outputs. The clauses may potentially assign different types to the same output, hence we need to take the join of types to match all of them (e.g. private and public values together comprise a private output).


\begin{figure*}

\begin{minipage}{0.39\linewidth}

\infrule[Lit]{\Gamma_j\vdash c_t :\ \tm{public}\ t}{}
\vspace{0.2cm}

\infrule[Var]{\Gamma_j\vdash (x :\ \tm{b})\ d\ t}{\Gamma_j \vdash x\ :\ d\ t}
\vspace{0.2cm}

\infrule[UnOp(\oplus)]{\Gamma_j\vdash\oplus(e) :\ d\ t}{\Gamma_j\vdash e\ :\ d\ t}
\vspace{0.2cm}

\infrule[BinOp(\oplus)]{\Gamma_j\vdash(e_1 \oplus e_2) :\ (d_1 \vee d_2)\ (t_1 \vee t_2)}{\Gamma_j\vdash e_1\ :\ d_1\ t_1 \quad \Gamma_j\vdash\ e_2\ :\ d_2\ t_2}

\end{minipage}
\begin{minipage}{0.59\linewidth}

\infrule[DB]{\Gamma\vdash p : [d_1,\ldots,d_n]\ [t_1,\ldots,t_n]}{\tm{:-type(}p(a_1\ :\ d_1\ t_1,\ldots,a_n\ :\ d_n\ t_n)\tm{)} \in \set{D_1,\ldots,D_m}}
\vspace{0.2cm}

\infrule[Goal]{\Gamma \vdash x_1 :\ d_1\ t_1 \ldots \Gamma \vdash x_k :\ d_k\ t_k}{G = \tm{?-}p((x_1\ :\ \tm{b} ) :\ d_1\ t_1,\ldots,(x_k\ :\ \tm{b})\ :\ d_k\ t_k,y_1\ :\ \tm{f},\ldots,y_{\ell}\ :\ \tm{f})}
\vspace{0.2cm}

\infrule[Cmp(\odot)]{\Gamma_j\vdash(e_1 \odot e_2)\ :\ (d_1 \vee d_2)\ \tbool}{\Gamma_j\vdash e_1\ :\ d_1\ t \quad \Gamma_j\vdash\ e_2\ :\ d_2\ t}
\vspace{0.2cm}

\infrule[Asgn]{\Gamma_j\vdash x\ :\ d\ t\quad\Gamma_j\vdash(x\ :\ \tm{f}\ =\ e)\ :\ \tm{public}\ \tbool}{\Gamma_j\vdash e\ :\ d\ t}

\end{minipage}
\vspace{0.2cm}

\infrule[Pred]{\Gamma_j\vdash p(e_1,\ldots,e_n) : \bigvee_{i \in I_b}(d_i \vee d'_i)\ \tbool\quad \forall i \in I_f:\ \Gamma_j\vdash x_i\ :\ d_i\ t_i}{\Gamma\vdash  p : [d_1,\ldots,d_n]\ [t_1,\ldots,t_n]\quad I_f=\set{i \in [n]\ |\ e_i = (x_i\ :\ \tm{f})}\quad I_b = [n]\setminus I_f\quad\forall i \in I_b:\ \Gamma_j\vdash e_i\ :\ d'_i\ t_i}
\vspace{0.2cm}

\infrule[Inputs]{\forall j \in [m]:\ \Gamma_j\vdash x_{j1}\ :\ d_1\ t_1\ldots\Gamma_j\vdash x_{jk}\ :\ d_k\ t_k}{\Gamma\vdash x_1\ :\ d_1\ t_1\ldots\Gamma\vdash x_k\ :\ d_k\ t_k\quad\forall j \in [m]:\ C_j = p(x_{j1}\ : \tm{b},\ldots,x_{jk}\ :\ \tm{b},\ldots) \imp q_j\lpend}
\vspace{0.2cm}

\infrule[Outputs]{\Gamma\vdash x_1\ :\ \bigvee_{j \in [m]} d_{j1}\ \bigvee_{j \in [m]} t_{j1}\ldots\Gamma\vdash x_k\ :\ \bigvee_{j \in [m]} d_{jk}\ \bigvee_{j \in [m]} t_{jk}}{\forall j \in [m]:\ \Gamma_j\vdash x_{j1}\ :\ d_{j1}\ t_{j1}\ldots\Gamma_j\vdash x_{jk}\ :\ d_{jk}\ t_{jk}\quad C_j = p(x_{j1}\ : \tm{b},\ldots,x_{jk}\ :\ \tm{b},\ldots) \imp q_j\lpend}
\caption{Type inference rules for a $\mu$-PrivaLog program $D_1\ \ldots\ D_m\ C_1\ \ldots\ C_n\ G$.}\label{fig:typeinfrules}
\end{figure*}




\section{Generation of a SecreC program}\label{sec:strans}

In this section, we describe how a typed $\mu$-PrivaLog program is transformed to a SecreC program. We describe the resulting program in terms of \emph{SecreC core language} described in Section~\ref{sec:secrec}.

The definition of transformation is given in Figure~\ref{fig:strans}. For easier readability, we define different transformation functions $\stransfunp : program \to program$, $\stransfunc : clause \to fun$, $\stransfunf : formula \to stmt$, $\stransfune : term \to e$.

\begin{figure*}
\begin{eqnarray*}
\stranse{\lpneg e} & := & \tm{!}\stranse{e}\\
\stranse{e_1 \oplus_i e_2} & := & \stranse{e_1}\ \stransop{\oplus}_i\ \stranse{e_2}\indnt\text{ for }\oplus=[\tm{+},\tm{-},\tm{*},\tm{/}, \lppow, \lpand, \lpor],\ \stransop{\oplus}=[\tm{+},\tm{-},\tm{*},\tm{/}, \lppow, \&,\ |\ ]\\
\stranse{e_1 \odot_i e_2} & := & \stranse{e_1}\ \stransop{\odot}_i\ \stranse{e_2}\indnt\text{ for }\odot=[\tm{<},\tm{=<},\tm{=:=},\tm{=/=},\tm{>=},\tm{=}],\ \stransop{\odot}=[\tm{<},\tm{<=},\tm{==},\tm{!==},\tm{>=},\tm{==}]\\
\stranse{c^t} & := & c^t\text{ for a constant }c^t\\
\stranse{x} & := & x\text{ for a variable }x
\end{eqnarray*}

\begin{eqnarray*}
\stransf{k,\ q@((x\ :\ \tm{f})\ \tm{=}\ e)} & := & x\tm{ : }\dom(x)\ \type(x)\ \tm{[]}\tm{ = }\stranse{e}\tm{;}\\
& & \tm{b}_{k}\tm{ : }\dom(q)\tm{ bool [] = true;}\\
\stransf{k,\ q@p(e_1,\ldots,e_n)} & := & \stransf{k_1, e_1\ \tm{=}\ \tm{T}_k[1]\ :\ \tm{b}}\ \tm{;} \ldots\ \tm{;}\ \stransf{k_n, e_n\ \tm{=}\ \tm{T}_k[n]\ :\ \tm{b}} \tm{;}\\
& & \tm{b}_k\tm{ : }\dom(q)\tm{ bool [] = b}_{k_1}\tm{ \& }\cdots\tm{ \& b}_{k_n}\tm{;}\\
\stransf{k,\ \lpneg q} & := & \stransf{k', q}\ \tm{;}\\
& & \tm{b}_k\tm{ : }\dom(q)\tm{ bool [] = !b}_{k'}\tm{;}\\
\stransf{k,\ q} & := & \tm{b}_{k}\tm{ : }\dom(q)\tm{ bool [] = }\stranse{q}\tm{;}\indnt\indnt\indnt\indnt\indnt\indnt\text{for any other }q
\end{eqnarray*}\\

$\stransc{j,\ p(x_1\ :\ \tm{b},\ldots,x_k\ :\ \tm{b},\ \ y_1\ :\ \tm{f},\ldots,y_{\ell}\ :\ \tm{f})\ \imp\ q@(q_1,\ldots,q_m)\ \lpend}$
\begin{eqnarray*}
& :=  & p_{j}(x_1\tm{ : } \dom(x_1)\ \type(x_1),\ldots,x_k\tm{ : } \dom(x_k)\ \type(x_k))\tm{ : } (\dom(q)\ \tm{bool}\ \tm{[]},\ \ \dom(y_1)\ \type(y_1)\ \tm{[]},\ldots,\dom(y_{\ell})\ \type(y_{\ell})\ \tm{[]})\\
& & \indnt\indnt(\tm{T}_1\ :\ \dom(t_1)\ \type(t_1),\ldots,\tm{T}_m :\ \dom(t_m)\ \type(t_m))\tm{ = \join(\getTable(}t_1\tm{)},\ldots,\tm{\getTable(}t_m\tm{)));}\\
& & \indnt\indnt\stransf{1, q_1} \ldots \stransf{m, q_m}\\
& & \indnt\indnt\tm{b : }\dom(q)\tm{ bool [] = b}_1\tm{ \& }\cdots\tm{ \& b}_m\tm{;}\\ 
& & \indnt\indnt\tm{return}\ \tm{(b, }y_1,\ldots,y_{\ell}\tm{);}\\
where & & t_i = p\indnt\indnt\text{ for }q_i = p(\ldots)\\
      & & t_i = \tm{DUAL}\indnt\text{ for non-predicative }q_i
\end{eqnarray*}

\begin{eqnarray*}
\stransg{\tm{?-}\ p((x_1\ :\ \tm{b})\ d_1\ t_1,\ldots,(x_k\ :\ \tm{b})\ d_k\ t_k,\ \ y_1\ :\ \tm{f},\ldots,y_{\ell}\ :\ \tm{f})\ \lpend}
\end{eqnarray*}
\begin{eqnarray*}
 & := & \tm{main()}\\
 & & \indnt x_1\tm{ : }d_1\ t_1\tm{ = \argument(1);}\ \ldots\ x_k\tm{ : }d_k\ t_k\tm{ = \argument(k);}\\
 & & \indnt\tm{(b : private bool [], }\ y_1\tm{ : private }\type(y_1)\ \tm{[]}\tm{,}\ldots\tm{,}y_\ell\tm{ : private }\type(y_{\ell})\ \tm{[]} \tm{)}\\
 & & \indnt\indnt \tm{ = \shuffle(\unique(\concat(}\ p_1(x_1,\ldots,x_k),\ldots,p_m(x_1,\ldots,x_k)\ \tm{)));}\\
 & & \indnt\tm{pb : public bool []  = \declassify(b);}\\
 & & \indnt\tm{\publish(\filter(pb, }y_1\tm{))}\tm{;}\ \ldots\ \tm{\publish(\filter(pb, }y_\ell\tm{))}\tm{;}\\ \\
\stransp{C_1\,\ldots\,C_m\ G} & := & \stransc{1,C_1}\ \ldots\ \stransc{m,C_m}\ \stransg{G}\\
\end{eqnarray*}
\caption{Transformation of a $\mu$-PrivaLog program to SecreC. The maps $\dom$ and $\type$ return the domain and the type of an expression as specified in the context $\Gamma$ generated by type inference rules of Fig.~\ref{fig:typeinfrules}. Here \tm{DUAL} is a dummy one-row table, which is basically ignored in a cross product, and $\unique(\vec{b},\vec{y}_1,\ldots,\vec{y}_{\ell}) = (\first(\vec{y}_1,\ldots,\vec{y}_{\ell})\ \tm{\&}\ \vec{b}, \vec{y}_1,\ldots,\vec{y}_{\ell})$.} \label{fig:strans} 
\end{figure*}

$\stransfunc$ converts each clause $C_j$ to a SecreC function. The bounded arguments of $C_j$ will be function inputs (denoted $x_i$), and free arguments will be function outputs (denoted $y_i$). The function first collects all extensional predicates called inside the rule body, and takes the cross product of corresponding database tables. This results in a number $m$ of potential non-deterministic inputs. It uses $\stransfunf$ and $\stransfune$ to convert the rule body to a sequence of SecreC statements. These statements compute two kinds of outputs. Firstly, there is a boolean vector $b$ of length $m$, which for every non-deterministic input says whether it satisfies the rule body. Secondly, there are vectors $y_1,\ldots,y_\ell$, each of length $m$, that contain valuations of variables $y_i$ for all $m$ potential solutions, including those that actually do not satisfy the rule. In general, $b$ is a private vector, so we do not see whether a particular solution is satisfiable or not.

$\stransfunp$ constructs the main function of SecreC program which computes solutions $(b_j, y_{j1},\ldots,y_{j\ell})$ using the SecreC functions that $\stransfunc$ has generated for all $j \in \set{1,\ldots,n}$, and concatenates $y_i = y_{1i} \| \ldots \| y_{ni}$ for all $i \in \set{1,\ldots,\ell}$ to get the bag of all possible solutions (some solutions may repeat). It also concatenates $b = b_1 \| \ldots \| b_n$ to keep track of satisfiability of these solutions. To convert a bag to a set, the solutions need to be deduplicated. This is done without revealing locations of duplicate solutions (due to privacy reasons), so a privacy-preserving deduplication does not remove any solutions, but just updates the private vector $b$ and sets satisfiability bit to $0$ for duplicate solutions. All solutions are then shuffled (together with $b$), and the shuffled bits $b$ are declassified, i.e. become visible to the computing parties. Shuffling conceals the initial order of $b$. Only those solutions for which $b = 1$ are published to the client. Semantics of the SecreC program is defined as the set of these solutions.


The SecreC program obtained from our running example (without additional optimizations) would look as follows.
\begin{small}
\begin{verbatim}
goal(Port : private str, CargoType : private str)
      : (private bool [], public str [], private float []) {

    (T1, T2, T3, T5, T6)
    = crossPR(
        getTable("ship"),
        getTable("port"),
        getTable("ship"),
        getTable("ship"),
        getTable("port"));

    //line 1
    Ship : public str [] = T1.name;
    b1 : private bool [] = (CargoType == T1.cargotype);

    //line 2
    b2 : private bool []     = (Port == T2.name);
    Capacity : public int [] = T2.capacity;

    //line 3
    b3   : private bool [] = (Ship == T3.name);
    CargoAmount : private int [] = T3.cargoamount;

    //line 4
    b4   : private bool [] = (Capacity >= CargoAmount);

    //line 5
    b5 : private bool []  = (Ship == T5.name);
    X1 : private float [] = T5.latitude;
    Y1 : private float [] = T5.longitude;
    Speed : public int [] = T5.speed;

    //line 6
    b6 : private bool  [] = (Port == T6.name);
    X1 : private float [] = T6.latitude;
    Y1 : private float [] = T6.longitude;

    //line 7
    b7 : public bool [] = true;
    Time = sqrt((X1* - X2*)^2 + (Y1* - Y2*)^2) / Speed;

    b = b1 & b2 & b3 & b4 & b5 & b6 7 b7;
    return(b, Ship, Time);
}

main () {
    port : private str      = argument(1);
    cargotype : private str = argument(2);
    (b : private bool [], Ship : public string [], Time : private float [])
       = shuffle(unique(goal(port, cargotype)));
    pb = declassify(b);
    publish(filter(bp, Ship));
    publish(filter(bp, Time));
}
\end{verbatim}
\end{small}

\subsection{Correctness}

\subsubsection{Validity of SecreC code.}

First of all, we need to prove that the syntax of generated SecreC program is valid.

\begin{lemma}[syntax validity]\label{lm:secrec:syntax}
Let $\mathbb{P}$ be a $\mu$-PrivaLog program. Then, the program $\stransp{\mathbb{P}}$ follows the syntax of SecreC (Figure~\ref{fig:secrec:syntax}).
\end{lemma}
\begin{proof}
We assume that the maps $\dom$ and $\type$ contain valid type labels that correspond to SecreC domain and type syntax respectively. It suffices to show that each transformation function $\stransfune$, $\stransfunf$, $\stransfunc$, $\stransfung$, $\stransfunp$ generate valid SecreC code.
\begin{itemize}
\item $\stranse{t}$ converts a term $t$ to a SecreC expression that corresponds to the syntax $e$ of SecreC. Validity is proven by induction on the structure of $t$. We assume that the variable name space $\vnames$ and constant value space $\vals$ are the same for PrivaLog and SecreC. Boolean and arithmetic operators of PrivaLog are adapted to the syntax of SecreC.
\item $\stransf{q}$ converts a formula $q$ into an assignment, which corresponds to the syntax of $stmt$.  For shortness, a declaration and an assignment are represented by a single line. Validity is proven by induction on the structure of $q$, and assuming that $\stransfune$ corresponds to the syntax $e$ of SecreC. Declarations are formally different from syntax of SecreC since we do not introduce the scope parentheses. Implicitly, a declaration covers all the subsequent statements, which we can treat as a $cons$ of several statements.
\item $\stransc{C}$ converts a clause $C$ to a SecreC function that satisfies syntax $fun$ of SecreC. The function header matches the syntax of header of $fun$. The first line incorporates several function calls and $m$ assignments, where for all $i \in [m]$ each assignment in turn corresponds to $\arity(t_i)$ assignments for each column of $t_i$. The second line corresponds to a sequence of SecreC $stmt$ statements by validity of $\stransfunf$. The remaining two lines are an assignment (which is also a $stmt$) and a return statement.
\item $\stransg{G}$ converts a goal $G$ to a main-function. The body of this function consists of declarations, assignments, and function calls, all of which are valid instances of  $stmt$.
\item $\stransp{\mathbb{P}}$ for $\mathbb{P} = D_1\ldots D_m\ C_1 \ldots C_n\ G$ generates functions $\stransc{1,C_1}\ \ldots\ \stransc{n,C_n}$ which correspond to syntax of $fun$, and the main function $\stransg{G}$ which corresponds to syntax of $\tm{main ()}\ stmt$. Altogether, this corresponds to the syntax of $prog$.
\end{itemize}
\end{proof}

In addition to syntax check, the SecreC program needs to pass a type check. While data type check is essential for correctness, the privacy domain check will guarantee privacy. The type check rules of SecreC (an adapted version of the monomorphic variant of rules taken from~\cite{UT:Randmets17}) is given in Figure~\ref{fig:secrec:typecheck}. The data type check is quite standard. The main idea behind privacy domain check is that the result of applying any operation on inputs with privacy domains $d_1,\ldots,d_n$ will be $d_1 \vee \cdots \vee d_n$, i.e. \tm{private} iff at least one of the input domains $d_i$ is \tm{private}, so that a private value will never be assigned to a public variable. The only way to do it explicitly is the function {\declassify}.

\begin{figure*}

\begin{minipage}{0.39\linewidth}

\infrule[Var]{P;\Gamma\vdash x\ :\ d\ t}{(x\ :\ d\ t)\in\Gamma}
\vspace{0.2cm}

\infrule[Skip]{P;\Gamma\vdash\tm{skip}}{}
\vspace{0.2cm}

\infrule[Cons]{P;\Gamma\vdash s_1\ \tm{;}\ s_2}{P;\Gamma\vdash s_1\quad P;\Gamma\vdash s_2}
\vspace{0.2cm}

\infrule[Decl]{P;\Gamma\vdash \set{x\ :\ d\ t\ \tm{;}\ s}}{\Gamma' = (\Gamma, x\ : d\ t)\quad P;\Gamma'\vdash s}
\vspace{0.2cm}

\infrule[UnOp]{P;\Gamma\vdash \oplus e\ :\ d\ t}{P;\Gamma\vdash e\ :\ d\ t}
\vspace{0.2cm}
\end{minipage}
\begin{minipage}{0.59\linewidth}

\infrule[Lit]{P;\Gamma\vdash c_t\ :\ \tm{public}\ t}{}
\vspace{0.2cm}

\infrule[Asgn]{P;\Gamma\vdash x\ \tm{=}\ e}{(x\ : d\ t) \in \Gamma\quad P;\Gamma\vdash e\ :\ d\ t}
\vspace{0.2cm}

\infrule[Return]{P;\Gamma\vdash \screturn\ e}{(\screturn :\ d\ t) \in \Gamma\quad P;\Gamma\vdash e\ :\ d\ t}
\vspace{0.2cm}

\infrule[BinOp]{P;\Gamma\vdash e_1 \oplus e_2\ :\ (d_1 \vee d_2)\ t}{P;\Gamma\vdash e_1\ :\ d_1\ t\quad P;\Gamma\vdash e_2\ :\ d_2\ t}
\vspace{0.2cm}

\infrule[Declassify]{P;\Gamma\vdash \declassify\ e\ :\ \tm{public}\ t}{P;\Gamma\vdash e\ :\ d\ t}
\vspace{0.2cm}
\end{minipage}

\infrule[Cmp(\odot)]{P;\Gamma\vdash e_1 \odot e_2\ :\ (d_1 \vee d_2)\ \tbool}{P;\Gamma\vdash e_1\ :\ d_1\ t \quad P;\Gamma\vdash\ e_2\ :\ d_2\ t}
\vspace{0.2cm}

\infrule[Call]{P;\Gamma\vdash f(e_1,\ldots,e_n)\ :\ d\ t}{P\vdash f : (d_1\ t_1,\ldots,d_n\ t_n) \to\ d\ t\quad P;\Gamma\vdash e_1\ :\ d_1\ t_1 \quad P;\Gamma\vdash\ e_n\ :\ d_n\ t_n}
\vspace{0.2cm}

\infrule[Function]{P\vdash f(x_1\ :\ d_1\ t_1,\ldots,x_n\ :\ d_n\ t_n)\ d\ t\ s}{P;(\tm{return}\ :\ d\ t,x_1\ :\ d_1\ t_1,\ldots,x_n\ :\ d_n\ t_n)\vdash s}
\vspace{0.2cm}

\infrule[Program]{P}{P,\emptyset\vdash \tm{body}(P)\quad \forall i \in [n]:\ P\vdash F_i}
\vspace{0.2cm}
\caption{Type check rules for a SecreC program $F_1\ \ldots\ F_n\ P$.}\label{fig:secrec:typecheck}
\end{figure*}

\begin{theorem}[type validity]\label{thm:secrec:type}
Let $\mathbb{P}$ be a $\mu$-PrivaLog program that successfully passed type inference of $\mu$-PrivaLog (Figure~\ref{fig:typeinfrules}). 
Then, the program $\stransp{\mathbb{P}}$ passes type check of SecreC (Figure~\ref{fig:secrec:typecheck}). 
\end{theorem}
\begin{proof}
First of all, we apply Lemma~\ref{lm:secrec:syntax} to claim that $\stransp{\mathbb{P}}$ does correspond to the syntax of SecreC, which allows to make the type check rules applicable.



Our goal is to show that the type check rule $\sce{Program}$ is satisfied. For this, we need to show that each function declaration $F_j = \stransf{j,C_j}$ satisfies the rule $\sce{Function}$, and the body of the \tm{main} function $\tm{body}(P)$ passes the type check as a statement.

The rule $\sce{Function}$ initializes a context $\Gamma_j$ with the types of inputs and the return statement for the function $\stransf{j,C_j}$. These types are taken from the SecreC function signature, and for each input $x$ the type is $\dom(x)\ \type(x)$, where $\dom$ and $\type$ are the domain and type derived by $\mu$-PrivaLog inference rules. To satisfy $\sce{Function}$, we need to show that the function body passes the type check. First of all, let us show that all its expressions satisfy SecreC expression type checking rules.

\begin{itemize}

\item The type check rules $\sce{Lit}$, $\sce{Var}$ $\sce{UnOp}$, $\sce{BinOp}$, $\sce{Cmp}$ need to be satisfied in the expressions generated by $\stranse{t}$. This is ensured by the type inference rules $\dle{Lit}$, $\dle{Var}$ $\dle{UnOp}$, $\dle{BinOp}$, $\dle{Cmp}$.

The rule $\dle{Var}$ is only defined for bounded variables, so we need to assume that $t$ is a ground term. Note that $\stranse{t}$ is called either for an assignment $x\ \tm{=}\ t$, or a formula $t$ that is neither an assignment nor a predicate. In $\mu$-PrivaLog, $t$ is ground in both such cases. If $x = \tm{T}_k[i]$, then we are dealing with a fresh equality $e_i\ \tm{=}\ \tm{T}_k[i]$ that is not a part of $\mathbb{P}$ and hence has not passed an explicit type inference. However, such equalities may only be generated for $p(e_1,\ldots,e_n)$, and the rules $\dle{Pred}$ and $\dle{DB}$ together ensure that $e_i$ has the same type as $\tm{T}_k[i]$.

In addition, $\sce{UnOp(\tm{!})}$ and $\sce{BinOp(\tm{\&})}$ need to be satisfied several times for boolean expressions over variables $\tm{b}_k$ which are not variables of $\mathbb{P}$. In all such cases, the arguments of \tm{!} and \tm{\&} are in turn SecreC variables $\tm{b}_k$, and for each such $k$, a statement $\stransf{k,q}$ has been called at some point. We see that $\stransf{k,q}$ always contains a declaration of $\tm{b}_k$, so the rule $\sce{Decl}$ has extended the context $\Gamma$ of these statements with $\tm{b}_k$.

\item The rule $\sce{Declassify}$ is not used inside $\stransf{j,C_j}$.
\end{itemize}

Now let us show that SecreC statement type checking rules are satisfied in all functions $\stransf{j,C_j}$.

\begin{itemize}
\item The rule $\sce{Skip}$ is trivially satisfied since there are no \tm{skip} statements.
\item The rule $\sce{Asgn}$ requires that the variable $x$ must have exactly the same type as the expression $e$ in an assignment $x \tm{=} e$.
     \begin{itemize}
     \item If $x$ is a PrivaLog variable, the rule $\dle{Asgn}$ ensures that $x$ has exactly the same type as $e$ for all expressions $e$ except $\tm{T}_k[i]$. If $e = \tm{T}_k[i]$, the type of $x$ is defined by the rule $\dle{Pred}$. The variables $\tm{T}_k[i]$ are declared in the first line of $\stransf{j,C_j}$, and their declared types are exactly the same.
     \item If $x = \tm{T}_k[i]$, we are dealing with the first line of $\stransf{j,C_j}$. If $t_k = \tm{DUAL}$, then the assignment is just omitted. Otherwise, $t_k = p$ for an EDB predicate $p$. These types have in turn been generated by the rule $\dle{DB}$, where they are taken directly from the schema declared by the user. The functions {\join} and {\getTable} return the cross product of EDB tables, and $e$ corresponds exactly to the columns of the table $t_k = p$, which have the same type as chosen by $\dle{DB}$
     \item If $x = \tm{b}_k$ for some $k$, it is assigned the domain $\dom(q)$ and the type $\type(q)$ which is $\tbool$ for a formula $q$. The rules $\dle{UnOp(\lpneg)}$, $\dle{BinOp(\lpand)}$, $\dle{BinOp(\lpor)}$ ensure type check for $\tm{b}_k$ computed as a combination of other $\tm{b}_{k'}$. If $q$ is a $\mu$-PrivaLog assignment, it always has \tm{public boolean} type by the rule $\dle{Asgn}$, and the constant $\tm{true}$ also has type \tm{public boolean}.
     \end{itemize}
\item The rule $\sce{Return}$ is satisfied because the declarations of output types in the function header are taken directly from $\Gamma$. We only need to verify that these types have indeed been declared in $\Gamma$. By assumption (which comes from standard Datalog), any variable $y_i$ that occurs in the rule header must also occur in the body. Since $y_i$ has label \tm{f} in the header, it will also  have label $\tm{f}$ in its first occurrence, which should be either an assignment or an EDB predicate, which results in adding a declaration of $y_i$ to $\Gamma$.
\item The rule $\sce{Decl}$ requires that all statements following a declaration must be satisfied assuming that we add the declared variables to the context $\Gamma$. We have proven in the previous points that all statements indeed pass the type check under these assumptions.
\item The rule $\sce{Cons}$ requires that each statement needs to pass the type check. We have proven in the previous points that all possible encountered statements do pass the type check one by one, so $\sce{Cons}$ is satisfied as well.
\end{itemize}

We have shown that all function declarations $F_i$ of a SecreC program $P$ satisfy the rule $\sce{Function}$. It remains to show that the body of $\stransg{G}$ passes the type check as well.

$\stransg{G}$ is a sequence of SecreC statements which needs to satisfy the rule $\sce{Cons}$. Let us show how its lines are satisfied one by one.
    \begin{itemize}
    \item The rule $\dle{Goal}$ ensures that the types of inputs $x_1,\ldots,x_k$ are the same as specified by the user in the PrivaLog program. The function $\tm{argument}$ will read variables of exactly these types.
    \item The types of the assignments on the second line look the same as the types of functions constructed by $\stransc{j,C_j}$. However, for each $j$, the type inference of $\mu$-PrivaLog belongs to a different context $\Gamma_j$ which may differ. The rule $\dle{Inputs}$ ensures that the input types of arguments $x_1,\ldots,x_n$ are exactly $d(x_1)\ t(x_1),\ldots,d(x_n)\ t(x_n)$ in context $\Gamma_j$ for all $j \in [n]$. To satisfy the  rule $\sce{Call}$, we also need that the return type of each $p_j$ matches the return type of a corresponding SecreC function. While $G$ does not set immediate type constraints on $p_j(x_1,\ldots,x_n)$, it does set a constraint on the output of {\concat}, {\unique}, {\shuffle} applied to these values. The function {\concat} concatenates together vectors of potentially different types, which requires that the result type should be $(d_1 \vee \ldots \vee d_n)\ (t_1 \vee \ldots \vee t_n)$. The function {\unique} returns a private boolean vector. The function {\shuffle} converts all public values to private, since otherwise positions of public values might reveal the shuffling permutation. We have $t(y_i) = t(y_{1i}) \vee \cdots \vee t(y_{ni})$ thanks to rule $\dle{Outputs}$.
    \item The type of \tm{pb} is \tm{public bool []}, and this line satisfies the rules $\sce{Declassify}$ and $\sce{Asgn}$.
    \item The last line publishes output to the client, and has no output. The input of {\publish} can be of any type. The first argument of {\filter} should have type \tm{public bool []}, and the second argument may be of any type. The rule $\sce{Call}$ is satisfied for this line. 
    \end{itemize}
All preconditions of the rule $\sce{Program}$ are thus satisfied for $\stransp{\mathbb{P}}$, so it passes the type check.
\end{proof}

\subsubsection{Equivalence of PrivaLog and SecreC semantics.}

In Figure~\ref{fig:semantics}, we have defined semantics of an arithmetic expression only for a single relation tuple (since the definitions have been borrowed from relational algebra). We will need to generalize the definition to relations. That is, for an expression $e$ and a set of relation tuples $\R$, we define expression semantics as
\[\sempar{e}\ \R = \set{\sempar{e}\ r\ |\ r \in \R}\enspace.\]
We provide a shorthand notation for a function that selects from a vector only those elements for which the satisfiability bit is $1$.
\begin{definition}[choice function for a vector]\label{def:choice}
Let $\vec{v} \in \vecset{m}$. Let $\vec{b}$ a boolean vector of length $m$. We define choice function as
\[\choice{}(\vec{v}, \vec{b}) = \set{v_i\ | i \in [m],\ b_i = true}\enspace.\]
\end{definition}
Generalization of Definition~\ref{def:choice} from $\vec{v} \in \vecset{m}$ to $[\vec{v}_1,\ldots,\vec{v}_n] \in \vecset{m}^n$ is straightforward. It can also be generalized to a program state as follows.
\begin{definition}[choice function for a state]\label{def:choice}
Let $s$ be a state over $\vecset{m}$. Let $\vec{b}$ a boolean vector of length $m$. Let $\set{x_1,\ldots,x_n} = Dom(s)$. We define choice function as
\[\choice{}(s, \vec{b}) = \set{s[x_1]_i,\ldots,s[x_n]_i\ |\ i \in[m],\ b_i = true}\enspace.\]
\end{definition}
In order to prove that SecreC and PrivaLog programs perform the same computation, we need to define an equivalence relation between SecreC program states and PrivaLog relations. First of all, if we consider a single variable, the only difference is that the elements in a vector are ordered, but in a relation they comprise an unordered set.
\begin{definition}[relation and vector equivalence w.r.t. attribute]\label{def:vecrelequiv}
A relation $\R$ and a vector $\vec{v} \in \vecset{m}$ are called \emph{equivalent} w.r.t. attribute $x$ (denoted $\vec{v} \eqvs{x} \R$) if
\[\set{v_i\ |\ i \in [m]} = \R[x]\enspace.\]
\end{definition}
Generalization of Definition~\ref{def:vecrelequiv} from $\vec{v} \in \vecset{m}$ to $[\vec{v}_1,\ldots,\vec{v}_n] \in \vecset{m}^n$ is straightforward. It can also be generalized to a program state as follows.
\begin{definition}[relation and program state equivalence w.r.t. schema]\label{def:staterelequiv}
A relation $\R$ and a program state $s$ over $\vecset{m}$ are called \emph{equivalent w.r.t. schema} $\Gamma = \set{x_1,\ldots,x_n}$ (denoted $s \eqvs{\Gamma} \R$)  if
\[\set{s[x_1]_i,\ldots,s[x_n]_i\ |\ i \in [m]} = \set{r[x_1],\ldots,r[x_n]\ | r \in \R}\enspace.\]
\end{definition}

Note that $s \eqvs{\Gamma} \R$ implies $s[x] \eqvs{x} \R$ for all $x \in \Gamma$. However, it does not hold the other way around, as the relations between variables are important for $s \eqvs{\Gamma} \R$.

We show that choice function preserves the equivalence relation.
\begin{proposition}\label{prop:choicesplit}
Let $s \in \vecset{m}$ be a state and $\R$, $\R'$ relations s.t. $\choice{}(s, \vec{b}) \eqvs{\Gamma} \R$ and  $\choice{}(s, \vec{b}') \eqvs{\Gamma} \R'$ for some schema $\Gamma$ and boolean vectors $\vec{b}$ and $\vec{b}'$ of length $m$. We have
\[
\choice{}(s, \vec{b} \wedge \vec{b}') \eqvs{\Gamma} \R \cap \R'\enspace.
\]
\end{proposition}
\begin{proof}
Let $\Gamma = \set{x_1,\ldots,x_n}$.
\begin{eqnarray*}
\choice{}(s, \vec{b} \wedge \vec{b}') \eqvs{\Gamma} \R \cap \R'
& \iff & \set{s[x_1]_i,\ldots,s[x_n]_i\ |\ i \in[m],\ b_{i} = true, b'_{i} = true}\\
& & = \set{r[x_1],\ldots,r[x_n]\ | r \in \R \cap \R')}\\
& \iff & \set{s[x_1]_i,\ldots,s[x_n]_i\ |\ i \in[m],\ b_{i} = true}\\
& & \cap \set{s[x_1]_i,\ldots,s[x_n]_i\ |\ i \in[m],\ b'_{i} = true}\\
& & = \set{r[x_1],\ldots,r[x_n]\ | r \in \R} \cap \set{r[x_1],\ldots,r[x_n]\ | r \in \R'}\\
& \Longleftarrow & \choice{}(s, \vec{b}) \eqvs{\Gamma} \R \wedge \choice{}(s, \vec{b}') \eqvs{\Gamma} \R'\enspace,
\end{eqnarray*}
where the latter statement holds by assumption.
\end{proof}



\begin{proposition}\label{prop:choicefilter}
Let $s \in \vecset{m}$ be a state s.t. $\choice{}(s, \vec{b}') \eqvs{\Gamma} \R$ for some schema $\Gamma$ and boolean vector $\vec{b}'$ of length $m$. Let $b$ a boolean expression defined over variables of $\Gamma$, and let $\vec{b} = \sempar{b}\ s$. We have
\[
\choice{}(s, \vec{b} \wedge \vec{b}') \eqvs{\Gamma} \sigma_{b} \R_{\vec{b}'}\enspace.
\]
\end{proposition}
\begin{proof}
Let $\Gamma = \set{x_1,\ldots,x_n}$. Since $\sempar{b}$ is only defined over $\Gamma$, we have $\sempar{b}\ s = \sempar{b}\ (s[x_1],\ldots,s[x_n])$.

\begin{eqnarray*}
\choice{}(s, \vec{b} \wedge \vec{b}') \eqvs{\Gamma} \sigma_{b} \R & \iff & \set{s[x_1]_i,\ldots,s[x_n]_i\ |\ i \in[m],\ b_i = true, b'_i = true}\\
& & = \set{r[x_1],\ldots,r[x_n]\ | r \in \sigma_{b}(\R)}\\
&\iff & \set{s[x_1]_i,\ldots,s[x_n]_i\ |\ i \in[m],\ (\sempar{b}\ s)_i = true, b'_i = true}\\
& & = \set{r[x_1],\ldots,r[x_n]\ | r \in \R, \sempar{b}\ r = true}\\
&\iff & \set{s[x_1]_i,\ldots,s[x_n]_i\ |\ i \in[m],\ \sempar{b}\ (s[x_1]_i,\ldots,s[x_n]_i) = true, b'_i = true}\\
& & = \set{r[x_1],\ldots,r[x_n]\ | r \in \R, \sempar{b}\ (r[x_1],\ldots,r[x_n]) = true}\\
&\Longleftarrow& \set{s[x_1],\ldots,s[x_n]\ |\ i \in [m], b'_i = true} = \set{r[x_1],\ldots,r[x_n]\ |\ r \in \R}\\
&\iff & \choice{}(s, \vec{b}')  \eqvs{\Gamma} \R\enspace,
\end{eqnarray*}
where the latter statement holds by assumption.
\end{proof}

Putting $\vec{b}' = \vec{true}$, we get the following corollary.
\begin{corollary}\label{cor:choicefilter}
Let $s$ be a state s.t. $s \eqvs{\Gamma} \R$ for some schema $\Gamma$. Let $b$ a boolean expression defined over variables of $\Gamma$, and let $\vec{b} = \sempar{b}\ s$. We have
\[
\choice{}(s, \vec{b}) \eqvs{\Gamma} \sigma_{b} \R\enspace.
\]
\end{corollary}

In the following, to avoid ambiguity, we denote semantics of Datalog as $\sempardl{\cdot}$, and semantics of SecreC as $\semparsc{\cdot}$. We will heavily use the properties of relational algebra listed in Proposition~\ref{prop:relalgebra}. 

\begin{lemma}\label{lm:filter}
Let $q$ be a ground PrivaLog formula that does not contain IDB or EDB predicates. Let $\Gamma$ be a schema satisfying $\Gamma \supseteq \vars(q)$ Then for any relation set $\R$, we have $\sempardl{q}_{\Gamma}(\R) = \sigma_{\sempardl{q}}(\univ_{\Gamma})$.
\end{lemma}
\begin{proof} We prove the statement by induction on size of $q$ and case distinction.
\begin{itemize}
\item Let $q = \tm{true}$.
      \begin{eqnarray*}
      \sempardl{\tm{true}}_{\Gamma}(\R) & = & \univ_{\Gamma}\\
      & = & \sigma_{true}(\univ_{\Gamma})\\
      & = & \sigma_{\sempardl{true}}(\univ_{\Gamma})
      \end{eqnarray*}
\item Let $q = \tm{false}$.
      \begin{eqnarray*}
      \sempardl{\tm{false}}_{\Gamma}(\R) & = & \emptyset\\
      & = & \sigma_{false}(\univ_{\Gamma})\\
      & = & \sigma_{\sempardl{false}}(\univ_{\Gamma})
      \end{eqnarray*}
\item Let $q = \lpneg q'$.
      \begin{eqnarray*}
      \sempardl{\lpneg q'}_{\Gamma}(\R) & = & \univ_{\Gamma} \setminus \sempardl{q'}_{\Gamma}(\R)\\
      & = & \univ_{\Gamma} \setminus \sigma_{\sempardl{q'}}(\univ_{\Gamma})\\
      & = & \sigma_{\neg\sempardl{q'}}(\univ_{\Gamma})\\
      & = & \sigma_{\sempardl{\lpneg q'}}(\univ_{\Gamma})
      \end{eqnarray*}
\item Let $q = q_1\ \lpand\ q_2$.
      \begin{eqnarray*}
      \sempardl{q_1\ \lpand\ q_2}_{\Gamma}(\R) & = & \sempardl{q_1}_{\Gamma}(\R) \cap \sempardl{q_2}_{\Gamma}(\R)\\
      & = & \sigma_{\sempardl{q_1}}(\univ_{\Gamma}) \cap \sigma_{\sempardl{q_2}}(\univ_{\Gamma})\\
      & = & \sigma_{\sempardl{q_1}\ \wedge \sempardl{q_2}}(\univ_{\Gamma})\\
      & = & \sigma_{\sempardl{q_1 \lpand q_2}}(\univ_{\Gamma})
      \end{eqnarray*}
\item Let $q = q_1\ \lpor\ q_2$.
      \begin{eqnarray*}
      \sempardl{q_1\ \lpor\ q_2}_{\Gamma}(\R) & = & \sempardl{q_1}_{\Gamma}(\R) \cup \sempardl{q_2}_{\Gamma}(\R)\\
      & = & \sigma_{\sempardl{q_1}}(\univ_{\Gamma}) \cup \sigma_{\sempardl{q_2}}(\univ_{\Gamma})\\
      & = & \sigma_{\sempardl{q_1} \vee \sempardl{q_2}}(\univ_{\Gamma})\\
      & = & \sigma_{\sempardl{q_1 \lpor q_2}}(\univ_{\Gamma})
      \end{eqnarray*}
\item $q = (e_1 \odot e_2)$ for $\oplus\in\set{\tm{<}, \tm{=<}, \tm{=:=}, \tm{=/=}, \tm{>=}, \tm{>}, \tm{=}}$
      \begin{eqnarray*}
      \sempardl{e_1\ \odot\ e_2}_{\Gamma}(\R) & = & \sigma_{\sempardl{e_1 \odot e_2}}(\univ_{\Gamma})
      \end{eqnarray*}
\end{itemize}
Since $q$ is ground, by assumption, it does not cover variable assignment. Hence, we have looked through all possible cases.
\end{proof}

We will prove equivalence of SecreC and PrivaLog programs on different granularity levels (expression, formula, clause, program). We start from expressions, showing how to increase the schema $\Gamma$ w.r.t. a state and a relation.

\begin{lemma}[equivalence of expression semantics]\label{lm:exprsem}
Let $e$ be an expression defined according to the syntax of Figure~\ref{fig:syntax}. Let $s \eqvs{\Gamma} \R$, $z \notin \Gamma$, and $\vars(e) \subseteq \Gamma$. We have
\[s[z \mapsto \semparsc{\stranse{e}}\ s]\ \eqvs{\Gamma \cup \set{z}}\ \sigma_{\sempardl{z =:= e}}\ (\R \times \univ_{z})\enspace\]
for $z \notin \schema(\R)$.
\end{lemma}
\begin{proof}
Let $s \in \vecset{m}$. Let $s' := s[z \mapsto \semparsc{\stranse{e}}\ s]$, and let $\R' := \sigma_{\sempardl{z =:= e}}\ (\R \times \univ_{z})$. Let $\Gamma = \set{x_1,\ldots,x_n}$. Suppose that
\begin{equation}\label{eq:exprsem}
\forall (v_1,\ldots,v_n) \in \vals^n:\ \semparsc{\stranse{e}}\ (v_1,\ldots,v_n) = \sempardl{e}\ (v_1,\ldots,v_n)\enspace.
\end{equation}
We have
\begin{eqnarray*}
\set{s'[x_1]_i,\ldots,s'[x_n]_i, s'[z]\ |\ i \in [m]} & = & \set{s[x_1]_i,\ldots,s[x_n]_i, (\semparsc{\stranse{e}}\ s)_i\ | i \in [m]}\\
& = & \set{s[x_1]_i,\ldots,s[x_n]_i, \semparsc{\stranse{e}}\ (s[x_1]_i,\ldots,s[x_n]_i)\ | i \in [m]}\\
\set{r[x_1],\ldots,r[x_n], r[z]\ |\ r \in \R'} & = & \set{r[x_1],\ldots,r[x_n], r[z]\ |\ \sempardl{z =:= e}\ r = true, r \in \R \times \univ_{z}}\\
 & = & \set{r[x_1],\ldots,r[x_n], r[z]\ |\ \sempardl{e}\ r = r[z], r \in \R \times \univ_{z}}\\
 & = & \set{r[x_1],\ldots,r[x_n], \sempardl{e}\ (r[x_1],\ldots,r[x_n])\ |\ r \in \R \times \univ_{z}}\\
 & = & \set{r[x_1],\ldots,r[x_n], \sempardl{e}\ (r[x_1],\ldots,r[x_n])\ |\ r \in \R}\enspace,
\end{eqnarray*}
and since $s \eqvs{\Gamma} \R$ these quantities are equal assuming that (\ref{eq:exprsem}) holds. It remains to prove (\ref{eq:exprsem}) for any expression $e$. We implicitly assume correspondence between PrivaLog and SecreC arithmetic and boolean operators (although formally hey are using different syntax). We prove the statement by induction on size of expression $e$.
\begin{eqnarray*}
\semparsc{\stranse{\lpneg e}}\ v & = & \semparsc{!\stranse{e}}\ v = \neg\ (\semparsc{\stranse{e}}\ v)\\
& = & \neg (\sempardl{e}\ v) = \sempardl{\lpneg e}\ v\\
\semparsc{\stranse{e_1 \oplus e_2}}\ v & = & \semparsc{\stranse{e_1}\ \stransop{\oplus}\ \stranse{e_2}}\ v\\
& = & \stransop{\oplus}(\semparsc{\stranse{e_1}}\ v, \semparsc{\stranse{e_2}}\ v)\\
& = & \oplus(\sempardl{e_1}\ v, \sempardl{e_2}\ v)= \sempardl{e_1\ \oplus\ e_2}\ v\\
& & \text{ for }\oplus\in\set{\tm{+},\tm{-},\tm{*},\tm{/},\lppow,\lpand,\lpor}\\
\semparsc{\stranse{e_1 \odot e_2}}\ v & = & \semparsc{\stranse{e_1}\ \stransop{\odot}\ \stranse{e_2}}\ v\\
& = & \stransop{\odot}(\semparsc{\stranse{e_1}}\ v, \semparsc{\stranse{e_2}}\ v)\\
& = & \odot(\sempardl{e_1}\ v, \sempardl{e_2}\ v)= \sempardl{e_1\ \odot\ e_2}\ v\\
& & \text{ for }\odot\in\set{\tm{<},\tm{=<},\tm{=},\tm{>=},\tm{>}}\\
\semparsc{\stranse{c_t}}\ v & = & c_t = \sempardl{c_t}\ v \text{ for a constant }c_t\\
\semparsc{\stranse{x_i}}\ (v_1,\ldots,v_n) & = & v_i = \sempardl{x}_i\ (v_1,\ldots,v_n)\text{ for a variable }x_i
\end{eqnarray*}
\end{proof}

For a formula, we want that the set of satisfiable relations of a PrivaLog rule is the same as the set of relations computed by the corresponding SecreC function. An important detail here is that each relation computed by the SecreC is accompanied by a satisfiability bit $b$, while in PrivaLog semantics such solutions are discarded. Hence, we prove that the sets of relations will be the same for those solutions $y_i$ where $b_i = 1$, and we do not care of solutions for which $b_i = 0$. Lemma~\ref{lm:formulasem} says that this property is maintained for each subsequence $q_1,\ldots,q_k$ of the PrivaLog rule body $q_1 \lpand \ldots \lpand q_m$, where $k \leq m$.

\begin{lemma}[equivalence of formula semantics w.r.t. input]\label{lm:formulasem}
\noindent Let $q = q_1\ \lpand \ldots \lpand q_m$ be a $\mu$-PrivaLog formula.

\noindent Let $\Gamma = \vars(q)$ be the variables of $q$.

\noindent Let $X = \Gamma \setminus \fvars(q)$, where $\fvars(q) = \set{x\ |\ x\ :\ \tm{f}\text{ occurs in }q}$.

\noindent Let $\R_X$ be a relation such that $\schema(\R_X) = X$.

\noindent Let $\R$ and $\R'$ be sets of relations such that, for all $k \in [m]$:
    \begin{itemize}
    \item $\R'_k[i] = \tm{T}_k[i]$;
    \item $\R'_1 \times \cdots \R'_m = \R_{j_1} \times \cdots \times \R_{j_m}$ where $q_k = p_{j_k}(\ldots)$ for all $k \in [m]$.
    \end{itemize}
\noindent Let $s_0$ be a state such that $s_0 \eqvs{X\cup \bigcup_{i=1}^m \set{\tm{T}_{i}[1],\ldots,\tm{T}_{i}[\arity(\R'_i)]}} (\R_X \times \prod_{i=1}^m \R'_i)$.

\noindent Let $\R^0 = \R_{X} \times \univ_{\Gamma \setminus X}$.

\noindent Let $\R^k = \sempardl{q_k}_{\Gamma}(\R)$, and let $s = \semparsc{\stransf{1, q_1}, \ldots, \stransf{m, q_m}}\ s_0$.

\noindent We have:
\[\choice{}(s, s[\tm{b}_1]\wedge \cdots \wedge s[\tm{b}_m]) \eqvs{\Gamma} (\R^0 \cap \R^1 \cap \cdots \cap \R^m)\enspace.\]
\end{lemma}
\begin{proof}
By induction on $k$, we prove a stronger statement. Let $s_k = \semparsc{\stransf{1, q_1}, \ldots, \stransf{k, q_k}}\ s_0$.  Let $\Gamma_0 = X$, and let $\Gamma_k := X \cup \vars(q_1\ \lpand \ldots\ \lpand q_k)$. For all $k \leq m$:
\[
\choice{}(s_k, s_k[\tm{b}_1] \wedge \cdots \wedge s_k[\tm{b}_k]) \eqvs{\Gamma_k} (\R^0 \cap \R^1 \cap \cdots \cap \R^k)\enspace.
\]
In the following, we will use notation $\mathcal{C}_k := \choice{}(s_k, s_k[\tm{b}_1] \wedge \cdots \wedge s_k[\tm{b}_k])$.

\noindent\textbf{Base:} We have $k = 0$, and $s = s_0$. By assumption, $s_0 \eqvs{\Gamma_0} \R_0$. By Corollary~\ref{cor:choicefilter},
    \begin{eqnarray*}
    \mathcal{C}_0 & = & \choice{}(s, true)\\
    & = & \choice{}(s_0, true)\\
    & \eqvs{\Gamma_0} & \sigma_{true}(\R^0) = \R^0\enspace.
    \end{eqnarray*}

\noindent\textbf{Step:} Let $k > 0$. Since each variable $\tm{b}_i$ is only assigned a value once in a function, we have $s_k[\tm{b}_i] = s_{k-1}[\tm{b}_i]$ for $i < k$. Hence,
\[\mathcal{C}_k = \choice{}(s_k, s_{k-1}[\tm{b}_1] \wedge \cdots \wedge s_{k-1}[\tm{b}_{k-1}] \wedge s_k[\tm{b}_k])\enspace.\]
Further simplification depends on the form of $q_k$. We need to apply case distinction.

\begin{itemize}
\item Let $q_k = \lpneg q$.
      \begin{eqnarray*}
      \R^k & = & \sempardl{\lpneg q}_{\Gamma}(\R) = \univ_{\Gamma} \setminus  \sempardl{q}_{\Gamma}(\R)\\
      s_k  & = & \semparsc{\stransf{k, \lpneg q}}\ s_{k-1}\\
           & = & \semparsc{\stransf{k', q}\ \tm{;}\ \tm{b}_k\tm{ : }d\tm{ = !b}_{k'}\ \tm{;}}\ s_{k-1}
      \end{eqnarray*}
      By induction hypothesis, $\mathcal{C}_{k'} \eqvs{\Gamma_{k'}} (\R^0 \cap \R^1 \cap \cdots \cap \R^{k-1} \cap \R^{k'})$. If $q$ is a ground term, then $s_k[x] = s_{k'}[x]$ for all variables $x$ except $b_{k'}$, but the variable $b_{k'}$ will not be used anywhere else and is not a part of $\Gamma$. We have $b_k = !b_{k'}$, hence,
      \begin{eqnarray*}
      C_k & = & \choice{}(s_k, s_k[\tm{b}_1] \wedge \cdots \wedge s_k[\tm{b}_{k-1}] \wedge s_k[\tm{b}_{k}])\\
      & = & \choice{}(s_k, s_{k'}[\tm{b}_1] \wedge \cdots \wedge s_{k'}[\tm{b}_{k-1}] \wedge !s_{k'}[\tm{b}_{k'}])\\
      & \eqvs{\Gamma} & (\R^0 \cap \R^1 \cap \cdots \cap \R^{k-1} \cap (\univ_{\Gamma} \setminus \R^{k'}))\\
      & \eqvs{\Gamma} & (\R^0 \cap \R^1 \cap \cdots \cap \R^{k-1} \cap \R^k)
      \end{eqnarray*}

      Since $q$ does contain free variables, it does not introduce any new variable declarations.

\item Let $q_k$ be any other ground formula without IDB/EDB predicates.

      \begin{eqnarray*}
      \R^k & = & \sempardl{q_k}_{\Gamma}(\R)\\
      s_k  & = & \semparsc{\stransf{k, q_k}}\ s_{k-1}\\
           & = & \semparsc{\tm{b}_k\tm{ : }d\tm{ bool [] = }\stranse{q_k}\tm{;}}\ s_{k-1}
      \end{eqnarray*}

       No new variables are introduced by $\semparsc{q_k}$ except $\tm{b}_k$, which are not a part of $\Gamma_k$. Therefore:

      \begin{itemize}
      \item $\Gamma_k = \Gamma_{k-1}$;
      \item $s_k[x] = s_{k-1}[x]$ for all $x \in \Gamma_k$;
      \item $\sempardl{q_k}_{\Gamma}(\R) = \sigma_{\sempardl{q_k}}(\univ_{\Gamma})$ by Lemma~\ref{lm:filter}.
      \end{itemize}

Applying induction hypothesis and Proposition~\ref{prop:choicefilter},
      \begin{eqnarray*}
      \mathcal{C}_k & = & \choice{}(s_k, s_{k-1}[\tm{b}_1] \wedge \cdots \wedge s_{k-1}[\tm{b}_{k-1}] \wedge s_k[\tm{b}_k])\\
      & = & \choice{}(s_{k-1}, s_{k-1}[\tm{b}_1] \wedge \cdots \wedge s_{k-1}[\tm{b}_{k-1}] \wedge \semparsc{\stranse{q_k}}\ s_{k-1})\\
      & \eqvs{\Gamma_{k-1}=\Gamma_k} & \sigma_{\sempardl{q_k}} (\R^0 \cap \cdots \cap \R^{k-1})\\
      & = & \sigma_{\sempardl{q_k}} \univ_{\Gamma} \cap (\R^0 \cap \cdots \cap \R^{k-1})\\
      & = & \R^k \cap (\R^0 \cap \cdots \cap \R^{k-1})\\
      & = & \R^0 \cap \cdots \cap \R^k\enspace.
      \end{eqnarray*}

\item Let $q_k = ((z : \tm{f})\ \tm{=}\ e)$.
      \begin{eqnarray*}
      \R^k & = & \sempardl{z = e}_{\Gamma}(\R)\\
          & = & \sigma_{\sempardl{z =:= e}} \univ_{\Gamma}\\
      s_k & = & \semparsc{\stransf{k,z\ \tm{=}\ e}}\ s_{k-1}\\
          & = & \semparsc{z\ \tm{:}\ d\ t\tm{=}\stranse{e}\tm{; b}_k\tm{:public bool [] = true;}}\ s_{k-1}
      \end{eqnarray*}

      In $\mu$-PrivaLog, the RHS of an assignment does not contain variables with label \tm{f}. Hence, each variable $x \in \vars(e)$ either has occurred before in $q_1,\ldots,q_{k-1}$, or $x \in X$, so $\vars(e) \subseteq \Gamma_{k-1}$. On the other hand, the label $\tm{f}$ of $z$ means that $z$ occurs first time in $q$, and by assumption $z \notin X$, so initialization of $z$ is a valid operation in SecreC. By induction hypothesis, $\choice{}(s_{k-1}, s_{k-1}[\tm{b}_1] \wedge \cdots \wedge s_{k-1}[\tm{b}_{k-1}]) \eqvs{\Gamma_{k-1}} (\R^0 \cap \cdots \cap \R^{k-1})$, so we can apply Lemma~\ref{lm:exprsem}. 
      \begin{eqnarray*}
      \mathcal{C}_k & = & \choice{}(s_k, s_{k-1}[\tm{b}_1] \wedge \cdots \wedge s_{k-1}[\tm{b}_{k-1}] \wedge s_k[\tm{b}_k])\\
      & = & \choice{}(s_k, s_{k-1}[\tm{b}_1] \wedge \cdots \wedge s_{k-1}[\tm{b}_{k-1}])\\
      & = & \choice{}(s_{k-1}[z \mapsto \semparsc{e}\ s_{k-1}], s_{k-1}[\tm{b}_1] \wedge \cdots \wedge s_{k-1}[\tm{b}_{k-1}])\\
      & = & \choice{}(s_{k-1}, s_{k-1}[\tm{b}_1] \wedge \cdots \wedge s_{k-1}[\tm{b}_{k-1}])[z \mapsto \semparsc{e}\ s_{k-1}]\\
      & \eqvs{\Gamma_{k-1} \cup \set{z} = \Gamma_k} & \sigma_{\sempardl{z=:=e}} (\R^0 \cap \cdots \cap \R^{k-1})\\
      & = & \sigma_{\sempardl{z=:=e}}\ \univ_{\Gamma} \cap (\R^0 \cap \cdots \cap \R^{k-1})\\
      & = & \R^k \cap (\R^0 \cap \cdots \cap \R^{k-1}) = \R^0 \cap \cdots \cap \R^k\enspace.
      \end{eqnarray*}

\item Let $q_k = p_j(z_1,\ldots,z_n)$.
      \begin{eqnarray*}
      \R^k & = & \sempardl{p_j(z_1,\ldots,z_n)}_{\Gamma}(\R) = \pi_{\Gamma}(\sigma_{\sempardl{z_1 =:= \R_j[1] \lpand \ldots \lpand z_n =:= \R_j[n]}}( \R_j \times \univ_\Gamma))\\
      s_k  & = & \semparsc{\stransf{k, p_j(z_1,\ldots,z_n)}}\ s_{k-1}\\
           & = & \semparsc{\stransf{k_1, z_1\tm{=}\ \tm{T}_k[1]\ :\ \tm{b}}\ \tm{;} \ldots \tm{;}\ \stransf{k_n, z_n\tm{=}\ \tm{T}_k[n]\ :\ \tm{b}}\ \tm{;}\ \tm{b}_k\tm{ : }d\tm{ = b}_{k_1}\tm{ \& }\cdots\tm{ \& b}_{k_n}\tm{;}}\ s_{k-1}
      \end{eqnarray*}
      Let $q_{k_i} := (z_i\tm{=}\ \tm{T}_k[i]\ :\ \tm{b})$ for $i \in [n]$. Consider the formula $q' := q_1\ \lpand \ldots\ \lpand q_{k-1}\ \lpand q_{k_1} \lpand \ldots \lpand q_{k_n}$. The variables $X' := \set{\tm{T}_k[i]\ | i \in [n]}$ occur for the first time and are \emph{not} a part of $\Gamma$. We have $\vars(q') = \Gamma \cup X'$, and $\Gamma\setminus\fvars(q') = X \cup X'$. The idea is to add $X'$ to $X$ and treat them as inputs.

Let $s' = \semparsc{\stransf{q'}}\ s_0$. By assumption, $s_0 \eqvs{X\cup \bigcup_{i=1}^m \set{\tm{T}_{i}[1],\ldots,\tm{T}_{i}[\arity(\R'_i)]}} (\R_X \times \prod_{i=1}^m \R'_i)$. By reordering variables, we can rewrite this as $s_0 \eqvs{(X\cup X') \cup \bigcup_{k\neq i=1}^m \set{\tm{T}_{i}[1],\ldots,\tm{T}_{i}[\arity(\R'_i)]}} (\R_X \times \R'_k) \times \prod_{k\neq i=1}^m \R'_i$, and can apply induction hypothesis to $q'$ with $\R_{X \cup X'} := \R_X \times \R'_k$. While the formula $q'$ is longer than $q$, we can still apply induction since the newly introduced expressions have a simpler structure for which we already have the proof. We get
      \begin{eqnarray*}
      \mathcal{C}'_k & := & \choice{}(s', s_{k-1}[\tm{b}_1] \wedge \cdots \wedge s_{k-1}[\tm{b}_{k-1}] \wedge s'[\tm{b}_{k_1}] \wedge \cdots \wedge s'[\tm{b}_{k_n}])\\
      & \eqvs{\Gamma_k} & \sigma_{\sempardl{z_1 = \tm{T}_j[1] \wedge \cdots \wedge z_n = \tm{T}_j[n]}} ((\R_X \times \R'_k \times \univ_{\Gamma\setminus X}) \cap (\R^1 \times \univ_{X'}) \cap \cdots \cap (\R^{k-1}\times \univ_{X'}))\\
      & = & \sigma_{\sempardl{z_1 = \tm{T}_k[1] \wedge \cdots \wedge z_n = \tm{T}_k[n]}} ((\R^0 \times \R'_k) \cap (\R^1 \times \univ_{X'}) \cap \cdots \cap (\R^{k-1}\times \univ_{X'}))\\
      & = & \sigma_{\sempardl{z_1 = \tm{T}_k[1] \wedge \cdots \wedge z_n = \tm{T}_k[n]}} ((\R^0 \cap \R^1 \cap \cdots \cap \R^{k-1}) \times \R'_k)\enspace.
      \end{eqnarray*}
      On the other hand,
      \begin{eqnarray*}
      \R^k & = & \pi_{\Gamma}(\sigma_{\sempardl{z_1 =:= \R_j[1] \lpand \ldots \lpand z_n =:= \R_j[n]}}( \R_j \times \univ_\Gamma))\enspace,
      \end{eqnarray*}
      hence
      \begin{eqnarray*}
      \mathcal{C}'_k & \eqvs{\Gamma_k} & \sigma_{\sempardl{z_1 = \tm{T}_k[1] \wedge \cdots \wedge z_n = \tm{T}_k[n]}} ((\R^0 \cap \R^1 \cap \cdots \cap \R^{k-1}) \times \R'_k)\\
      & = & \sigma_{\sempardl{z_1 = \R_j[1] \wedge \cdots \wedge z_n = \R_j[n]}} ((\R^0 \cap \R^1 \cap \cdots \cap \R^{k-1}) \times \R_j)\\
      & = & \pi_{\Gamma}(\sigma_{\sempardl{z_1 = \R_j[1] \wedge \cdots \wedge z_n = \R_j[n]}} ((\R^0 \cap \R^1 \cap \cdots \cap \R^{k-1}) \times \R_j))\\
      & = & \pi_{\Gamma}(\sigma_{\sempardl{z_1 = \R_j[1] \wedge \cdots \wedge z_n = \R_j[n]}} (\univ_{\Gamma} \times \R_j))\cap (\R^0 \cap \R^1 \cap \cdots \cap \R^{k-1})\\
      & = & \R^k \cap (\R^0 \cap \R^1 \cap \cdots \cap \R^{k-1})
      \end{eqnarray*}

      The only difference between states $s'$ and $s$ is that $\tm{b}_k$ is not defined in $s'$. Other than this, we have $s'[x] = s_k[x]$ for all $x \in \Gamma_k$. We have
      \begin{eqnarray*}
      \mathcal{C}_k & = & \choice{}(s_k, s_{k-1}[\tm{b}_1] \wedge \cdots \wedge s_{k-1}[\tm{b}_{k-1}] \wedge s_k[\tm{b}_{k}])\\
      & = & \choice{}(s_k, s_{k-1}[\tm{b}_1] \wedge \cdots \wedge s_{k-1}[\tm{b}_{k-1}] \wedge s_k[\tm{b}_{k1}] \wedge \cdots \wedge s_k[\tm{b}_{kn}])\\
      & = & \choice{}(s', s_{k-1}[\tm{b}_1] \wedge \cdots \wedge s_{k-1}[\tm{b}_{k-1}] \wedge s'[\tm{b}_{k1}] \wedge \cdots \wedge s'[\tm{b}_{kn}])\\
      & = & \mathcal{C}'_k = \R^0 \cap \cdots \cap \R^k\enspace.
      \end{eqnarray*}

\end{itemize}
\end{proof}

Lemma~\ref{lm:clausesem} says that, whenever a SecreC function $\goal_j$ produces a set of solutions $(\vec{w}_1,\ldots,\vec{w}_\ell)$ on inputs $v_1,\ldots,v_k$, the set of solutions for which $b_i = 1$ is exactly the same as the set of relations satisfying the clause $C_j(x_1,\ldots,x_n) \imp q$ in $\mu$-PrivaLog, where $(x_1 = v_1) \wedge \ldots \wedge (x_k = v_k)$.

\begin{lemma}\label{lm:clausesem}
Let $C_j := p_j(x_1\ :\ \tm{b},\ldots,x_k\ :\ \tm{b},\ \ y_1\ :\ \tm{f},\ldots,y_\ell\ :\ \tm{f})\ \imp\ q_1,\ldots,q_m$ be a $\mu$-PrivaLog clause defined according to the syntax of Figure~\ref{fig:muplsyntax}.


\noindent Let $\Gamma = \vars(q_1,\ldots,q_m)$.

\noindent Let $\R = (\R_1,\ldots,\R_n)$. 

\noindent Let $\tables : \vnames \to \vals^{n_1 \times n_2}$ be a function such that $\R_k = \set{p_k(\tables(p_k)[i,1],\ldots,\tables(p_k)[i,n_2])\ |\ i \in n_1}$.

\noindent Let $(\vec{b}, \vec{w}_1,\ldots,\vec{w}_\ell) = \semparsc{\stransc{j,C_j}}\ \funset{\tables}\ [v_1,\ldots,v_k]$, where $\funset{\tables}$ is the set of basic SecreC functions (Table~\ref{tbl:bb}).

\noindent We have
\begin{equation*}
\set{(w_{1 i},\ldots,w_{\ell i})\ |\ b_i = 1} = \pi_{y_1,\ldots,y_{\ell}}(\sigma_{x_1 = v_1,\ldots,x_k = v_k}\univ_{\Gamma} \cap \sempardl{q_1}_{\Gamma}(\R) \cap \cdots \cap \sempardl{q_m}_{\Gamma}(\R))\enspace.
\end{equation*}
\end{lemma}

\begin{proof}
Let $v_1,\ldots,v_k$ be the valuations of $x_1,\ldots,x_k$. Semantics of a SecreC function is defined as
\[\semparsc{\stransc{j,C}}\ \funset{\tables}\ [v_1,\ldots,v_k] := (\semparsc{S}\ \funset{\tables}\ \set{x_1 \mapsto v_1,\ldots,x_k \mapsto v_k})\ [\tm{b},y_1,\ldots,y_\ell]\enspace,\]
where $S$ is a sequence of SecreC statements defined as in Figure~\ref{fig:strans}. Let $s' = \set{x_1 \mapsto v_1,\ldots,x_k \mapsto v_k}$ be the initial program state, and let $X = \set{x_1,\ldots,x_k}$. We have $s' \eqvs{X} \sigma_{x_1 = v_1,\ldots,x_k = v_k}\univ_{\Gamma}$. The first line of $S$ uses blackbox functions {\join} and {\getTable} to compute the cross product $\R'_1 \times \cdots \R'_m = \R_{j_1} \times \cdots \times \R_{j_m}$, where $q_k = p_{j_k}(\ldots)$ for all $k \in [m]$. Technically, the inputs $s'[x_i]$ are scaled to the length of cross product table, so that operations on vectors would be correct. This results in a new state $s_0$ satisfying $s_0 \eqvs{X\cup \bigcup_{i=1}^m \set{\tm{T}_{i}[1],\ldots,\tm{T}_{i}[\arity(\R'_i)]}} (x_1,\ldots,x_k) \times \prod_{i=1}^m \R'_i$. Let $\R_X := \set{(x_1,\ldots,x_k)}$. The quantities $q$, $\Gamma$, $X$, $\R$, $\R_X$ and $s_0$ satisfy the preconditions of Lemma~\ref{lm:formulasem}, which says that
\[\choice{}(s, s[\tm{b}_1] \wedge \cdots s[\tm{b}_m]) \eqvs{\Gamma} (\R^0 \cap \cdots \cap \R^m)\enspace,\]
where $\R^k = \sempardl{q_k}_{\Gamma}(\R)$ for $k > 0$, and $\R^0 = \sigma_{x_1 = v_1,\ldots,x_k = v_k}\univ_{\Gamma}$. Since $\set{y_1,\ldots,y_{\ell}} \subseteq \Gamma$, we also have
\[\choice{}(s, s[\tm{b}_1] \wedge \cdots s[\tm{b}_m]) \eqvs{y_1,\ldots,y_{\ell}} (\R^0 \cap \cdots \cap \R^m)\enspace.\]

The last line of $S$ defines $s[\tm{b}] = s[\tm{b}_1] \wedge \cdots \wedge s[\tm{b}_m]$. We get
\[\choice{}((\vec{w}_1,\ldots,\vec{w}_{\ell}), \vec{b}) = \choice{}(s, s[\tm{b}])[y_1,\ldots,y_\ell] = \choice{}(s, s[\tm{b}_1] \wedge \cdots s[\tm{b}_m])[y_1,\ldots,y_\ell]\enspace,\]
which implies
\[\choice{}((\vec{w}_1,\ldots,\vec{w}_{\ell}), \vec{b}) \eqvs{y_1,\ldots,y_{\ell}} (\sigma_{x_1 = v_1,\ldots,x_k = v_k}\univ_{\Gamma} \cap \sempardl{q_1}_{\Gamma}(\R) \cap \cdots \cap \sempardl{q_m}_{\Gamma}(\R))\enspace,\]
which is equivalent to
\[\set{(w_{1 i},\ldots,w_{\ell i})\ |\ b_i = 1} = \pi_{y_1,\ldots,y_{\ell}}(\sigma_{x_1 = v_1,\ldots,x_k = v_k}\univ_{\Gamma} \cap \sempardl{q_1}_{\Gamma}(\R) \cap \cdots \cap \sempardl{q_m}_{\Gamma}(\R))\enspace.\]
\end{proof}

We are now ready to prove the main correctness theorem which says that a PrivaLog program $\mathbb{P}$ outputs the same set of relations as the corresponding SecreC program $\stransp{\mathbb{P}}$ when executed on the same inputs and the same database instance.
\begin{theorem}\label{thm:correctness}
Let $\mathbb{P}$ be a $\mu$-PrivaLog program. We have
\[\semparsc{\stransp{\mathbb{P}}} = \sempardl{\mathbb{P}}\enspace.\]
\end{theorem}
\begin{proof}
Let $\mathbb{P} = D_1\,\ldots\,D_m\ C_1\,\ldots\,C_n\ G$. Let $G = \tm{?-}p(a_1,\ldots,a_k,y_1,\ldots,y_\ell)$. Let $v_1,\ldots,v_k$ be some valuations of inputs for $\mathbb{P}$ and $\stransp{\mathbb{P}}$, and let $\tables$ be the database. Let $\R = (\bot,\R_1,\ldots,\R_m)$, where
\[\R_k = \set{p_k(\tables(p_k)[i,1],\ldots,\tables(p_k)[i,n_2])\ |\ i \in n_1}\]
for all $k \in \set{1,\ldots,m}$. Semantics of a $\mu$-PrivaLog program $\mathbb{P}$ is defined as
\begin{equation}
\sempardl{\mathbb{P}} = \bigcup_{j \in \set{1,\ldots,n}} \sempardl{C_j\set{x_{j1} \mapsto v_1,\ldots,x_{jk} \mapsto v_k}}(\R)[0]\enspace,
\end{equation}
where $x_{j1},\ldots,x_{jk}$ are the first $k$ arguments of $C_j$.

On the other hand, a transformed program is defined as
\begin{eqnarray*}
\stransp{C_1\,\ldots\,C_n\ G} & := & \stransc{1,C_1}\ldots\stransc{n,C_n}\ \stransg{G}\enspace.
\end{eqnarray*}

The proof follows the definition of transformation $\stransfunp$ given in Figure~\ref{fig:strans}. Execution of the main function $\stransg{G}$ of SecreC program starts from an empty program state $s_0 = \emptystate$. On the first line, the function \textsf{argument} reads values $v_i$ provided by the client for inputs $x_i$. This results in program state $s_1 = \set{x_1 \mapsto v_1,\ldots,x_k \mapsto v_k}$.

On the second line, SecreC functions $\goal_j = \semparsc{\stransc{j,C_j}}$ are called for $j \in \set{1,\ldots,n}$. Each function gets arguments $s_1[x_1] = v_1,\ldots,s_1[x_k] = v_k$, and returns some outputs $(\vec{b}_j, \vec{w}_{j1},\ldots,\vec{w}_{j\ell})$. Applying Lemma~\ref{lm:clausesem}, for each SecreC function constructed from a clause $p_j(z_1,\ldots,z_n) \imp q_{j1}\ \lpand \ldots\ \lpand q_{jm}$ with $\vars(q_{j1},\ldots,q_{jm}) = \Gamma$, we get
\begin{equation*}
\set{(w_{j1i},\ldots,w_{j\ell i})\ |\ b_i=1} = \pi_{y_1,\ldots,y_\ell}(\sigma_{(x_i = v_i)_{i \in [k]}}\univ_{\Gamma} \cap \sempardl{q_{j1}}_{\Gamma}(\R) \cap \cdots \cap \sempardl{q_{jm}}_{\Gamma}(\R))\enspace.
\end{equation*}

For each $i \in \set{1,\ldots,\ell}$, the function {\concat} constructs a state $s_2$ such that $s_2[y_i]$ is concatenation of the vectors $\vec{w}_{ji}$ for $j \in \set{1,\ldots,n}$, and $s_2[\tm{b}]$ is concatenation of $\vec{b}_j$ for $j \in \set{1,\ldots,n}$. The function {\unique} marks down duplicate solutions. The function {\shuffle} changes the order in which the solutions come, but keeps relations between $\vec{w}_{ji}$ and $\vec{b}_j$. This results in a state $s_2$ satisfying
\begin{equation*}
\set{(s_2[y_1]_i,\ldots, s_2[y_\ell]_i)\ |\ s_2[\tm{b}]_i} = \\ \bigcup_{j \in \set{1,\ldots,n}}\sigma_{(x_i = v_i)_{i \in [k]}}\univ_{\Gamma} \cap \sempardl{q_{j1}}_{\Gamma}(\R) \cap \cdots \cap \sempardl{q_{jm}}_{\Gamma}(\R)\enspace.
\end{equation*}

The next line declassifies private values $s_2[\tm{b}]$ by writing them to public $s_3[\tm{pb}]$. We have $s_3[\tm{pb}] = s_2[\tm{b}]$ and $s_3[y_i] = s_2[y_i]$ for all $i \in [\ell]$. The last line publishes to the client solutions $s_3[y_i]$ for which the declassified bit is $1$. That is, for each rule $j \in \set{1,\ldots,n}$, the client receives vectors that correspond to relation
\begin{eqnarray*}
R & = & \set{(s_3[y_1]_k,\ldots, s_3[y_\ell]_k)\ |\ s_3[\tm{pb}]_k}\\
& = & \bigcup_{j \in \set{1,\ldots,n}} \pi_{y_1,\ldots,y_\ell}(\sigma_{(x_i = v_i)_{i \in [k]}}\univ_{\Gamma} \cap \sempardl{q_{j1}}_{\Gamma}(\R) \cap \cdots \cap \sempardl{q_{jm}}_{\Gamma}(\R))\\
& = & \bigcup_{j \in \set{1,\ldots,n}} \pi_{y_1,\ldots,y_\ell}\sigma_{(x_i = v_i)_{i \in [k]}} \sempardl{C_j}(\R)\\
& = & \bigcup_{j \in \set{1,\ldots,n}} \sempardl{C_j\set{x_i \mapsto v_i}}(\R)[0]\enspace.
\end{eqnarray*}
The set of published answers is the same as for the $\mu$-PrivaLog program $\mathbb{P}$.
\end{proof}

\subsection{Security}
Let us now consider security of the resulting SecreC program. We need to estimate how much is leaked to the client who makes the query, as well as to Sharemind computing parties. Thanks to properties of SecreC, we know that the computing parties do not learn anything except what is deliberately declassified (which is the shuffled satisfiability bit $b$), and the client does not learn anything except what is deliberately published (which is the set of satisfying solutions). Informally, we are going to prove the following:
\begin{enumerate}
\item Computing parties learn only the size of the database, and the number of satisfying solutions for $\mathbb{P}$.
\item The client learns only the satisfying solutions for $\mathbb{P}$.
\end{enumerate}

Security of the resulting SecreC program is stated in Theorem~\ref{thm:security}. We base the security definition on the non-interference properties of two programs sharing common pre-agreed side-information~\cite{noninterf}. Let $x\ \sim\ x'$ denote that the variables $x$ and $x'$ have the same probability distribution.

\begin{lemma}\label{lm:security}
Let $\mathbb{P}$ be a PrivaLog program with input variables $[x_1,\ldots,x_k]$ and output variables $[y_1,\ldots,y_{\ell}]$. Let $s$ be the final state of the execution $\sempar{\stransp{\mathbb{P}}}\ [v_1,\ldots,v_k]\ \tables$, and $s'$ the final state of the execution $\sempar{\stransp{\mathbb{P}}}\ [v'_1,\ldots,v'_k]\ \tables'$, where $\tables$ and $\tables'$ have the same size. We have:
\begin{enumerate}
\item $s[pb] \sim s'[pb]$ for $|\sempar{\mathbb{P}}\ [v_1,\ldots,v_k]\ \tables| = |\sempar{\mathbb{P}}\ [v'_1,\ldots,v'_n]\ \tables'|$.
\item $s[y_1] \cdot s[pb],\ldots,s[y_\ell] \cdot s[pb] \sim s'[y_1] \cdot s'[pb],\ldots,s'[y_\ell] \cdot s'[pb]]$ for $\sempar{\mathbb{P}}\ [v_1,\ldots,v_k]\ \tables = \sempar{\mathbb{P}}\ [v'_1,\ldots,v'_n]\ \tables'$.
\end{enumerate}
\end{lemma}
\begin{proof}
We prove the statements one by one.
\begin{enumerate}
\item Let $|\sempar{\mathbb{P}}\ [v_1,\ldots,v_k]\ \tables| = |\sempar{\mathbb{P}}\ [v'_1,\ldots,v'_n]\ \tables'|$. By Theorem~\ref{thm:correctness}, we also have $|\sempar{\stransp{\mathbb{P}}}\ [v_1,\ldots,v_k]\ \tables| = |\sempar{\stransp{\mathbb{P}}}\ [v'_1,\ldots,v'_n]\ \tables'|$. The final output is obtained after leaving only such solutions $s[y_i]_k$ that $s[b]_k = 1$. The number of such solutions is the same $m$ for $s$ and $s'$.

On the other hand, consider the value $s[pb]$. This vector is obtained after shuffling $s[b]$. The function $\shuffle$ of SecreC performs a uniform shuffle, so we have $s[pb] \sim \shuffle([\overbrace{1,\ldots,1}^{m},0,\ldots,0])$. The total number of elements in $s[pb]$ depends only on the database size, as in $\stransp{\mathbb{P}}$ we have not performed any branching based on private data. The database size is the same for $s$ and $s'$, hence, we have $s[pb] \sim \shuffle([\overbrace{1,\ldots,1}^{m},0,\ldots,0]) \sim s'[pb]$.

\item Let $\sempar{\mathbb{P}}\ [v_1,\ldots,v_k]\ \tables = \sempar{\mathbb{P}}\ [v'_1,\ldots,v'_n]\ \tables'$. By Theorem~\ref{thm:correctness}, we also have $\sempar{\stransp{\mathbb{P}}}\ [v_1,\ldots,v_k]\ \tables = \sempar{\stransp{\mathbb{P}}}\ [v'_1,\ldots,v'_n]\ \tables'$. The final output is obtained after leaving only such solutions $s[y_i]_k$ that $s[b]_k = 1$. The set of such solutions is the same $(y_{1k},\ldots,y_{\ell k})_{k \in [m]}$ for $s$ and $s'$.

On the other hand, consider the value $s[y_1] \cdot s[pb],\ldots,s[y_\ell] \cdot s[pb]$. Similarly to the previous point, since function {\shuffle} of SecreC performs a uniform shuffle, we have $s[y_i] \cdot s[pb] \sim \shuffle([(y_{ik})_{k \in [m]},0,\ldots,0]) \sim s'[y_i] \cdot s'[pb]$ for all $i \in [\ell]$. \qedhere
\end{enumerate}
\end{proof}

It suffices to show that non-interference on $pb$ is sufficient for privacy against computing parties, and non-interference on $y_1 \cdot pb_1,\ldots,y_n \cdot pb_n$ is sufficient for privacy against computing parties. Let $s\ \downarrow\ P$ denote the set of values of all variables that are visible to the party $P$.
\begin{theorem}[Security]\label{thm:security}
Let $\mathbb{P}$ be a PrivaLog program with input variables $[x_1,\ldots,x_k]$ and output variables $[y_1,\ldots,y_{\ell}]$. Let $s$ be the final state of the execution $\sempar{\stransp{\mathbb{P}}}\ [v_1,\ldots,v_k]\ \tables$, and $s'$ the final state of the execution $\sempar{\stransp{\mathbb{P}}}\ [v'_1,\ldots,v'_k]\ \tables'$. After executing $\stransp{\mathbb{P}}$ on Sharemind with computing parties $S_1, S_2, S_3$ and a client party $C$:
\begin{enumerate}
\item $s \downarrow\ S_i \sim s' \downarrow\ S_i$ for all $i \in \set{1,2,3}$, if
    \begin{itemize}
    \item $|\sempar{\mathbb{P}}\ [v_1,\ldots,v_n]\ \tables| = |\sempar{\mathbb{P}}\ [v'_1,\ldots,v'_n]\ \tables'|$;
    \item the values of public variables are the same in $v_i$ and $v'_i$, $\tables$ and $\tables'$.
    \end{itemize}
\item $s \downarrow\ C \sim s' \downarrow\ C$ for $\sempar{\mathbb{P}}\ [v_1,\ldots,v_n]\ \tables = \sempar{\mathbb{P}}\ [v'_1,\ldots,v'_n]\ \tables'$.
\end{enumerate}
\end{theorem}
\begin{proof}
We prove the statements one by one.
\begin{enumerate}
\item The type system of SecreC ensures that only the variables that are labeled \tm{public} are visible to $S_i$. In $\stransp{\mathbb{P}}$, there are many more public variables than just $pb$. Let us list all of them, and show that the values of the other variables do not depend on private inputs.
\begin{itemize}
\item Public constants are a part of the program, and they are always the same for $s$ and $s'$.
\item There are more public variables that get their values from an assignments. Type system of SecreC ensures that only public data can be written into public variables. In $\stransp{\mathbb{P}}$, the public inputs have by assumption exactly the same values in $s$ and $s'$, and the private inputs are assigned private data type. The only way to convert a private variable to public is to use a special function {\declassify}, and the only public variable of $\stransp{\mathbb{P}}$ that gets its value in this way is $pb$. For $pb$, the claim follows directly from Lemma~\ref{lm:security}.
\item While there are no other {\declassify} calls in $\stransp{\mathbb{P}}$, the libraries of SecreC actually do contain some instances of {\declassify}. For the SecreC functions with {\declassify} that have been used in this paper (in particular, count sort~\cite{DBLP:conf/nordsec/BogdanovLT14}), non-interference has already been proven, and we have $s[z] \sim s'[z]$ for all such variables $z$.
\end{itemize}
\item In SecreC, the only function that delivers data to the client party $C$ is the function {\publish}. In $\stransp{\mathbb{P}}$, this function is not used in any of the imported libraries, and the only place where it occurs is the line $\tm{\publish(\filter(pb, }y_i\tm{))})_{i \in [\ell]}$. This opens to the client exactly the values $s[y_i]_k$ for which $s[b]_k = 1$, and there are no more values opened. These values comprise a subset of $s[y_1] \cdot s[pb],\ldots,s[y_\ell] \cdot s[pb]$. We have $s \downarrow C \subseteq s[y_1] \cdot s[pb],\ldots,s[y_\ell] \cdot s[pb]$, and the claim then follows directly from Lemma~\ref{lm:security}. \qedhere
\end{enumerate}
\end{proof}

\section{Implementation and Evaluation}\label{sec:eval}

We have implemented full PrivaLog$\to$SecreC transformation in Haskell language. The implementation is available in GitHub repository \url{https://github.com/fs191/secure-logic-programming}. The results of this paper correspond to branch \texttt{interactive}.

Haskell tool translates a PrivaLog program to a SecreC program, and an instance of Sharemind platform is needed to run the produced program. The platform is available at \url{https://sharemind.cyber.ee/sharemind-mpc/}. The Haskell tool is compatible with both the emulator and the licensed version. The benchmarks of this paper have been run on March 2019 release of Sharemind.

\subsection{Extensions and Optimizations}

In this section, we describe some extensions and optimizations that are not covered by the basic transformation.


\paragraph{Putting disjunctions back.}
During preprocessing, we get rid of all disjunctions, splitting a single rule $p \imp q_1 \lpor q_2$ into $p \imp q_1$ and $p \imp q_2$. This may lead to repeating code in different rules. Hence, after the type inference is complete and we have obtained a $\mu$-PrivaLog program, we start merging rules back. The rules of the form $p \imp q \lpand q_1$ and $p \imp q \lpand q_2$ can be merged into $p \imp q \lpand (q_1 \lpor q_2)$ only if $q_1$ and $q_2$ are both ground, so we do not introduce non-determinism into the rules, and $q_1 \lpor q_2$ can be treated similarly to a negation.

\paragraph{Operations on string data.} Since the only operation on strings that we are using so far is comparison, we compute (in a privacy-preserving way) CRC32 hash of each string, and operate on integers. We note that, while CRC32 hash is quite good at avoiding collisions, it still does not get rid of them completely. In critical applications, we may still need to compare strings elementwise.

\paragraph{Aggregations.}
So far, we have not mentioned aggregations in the syntax and semantics of PrivaLog, and we assumed that there are no high order predicates. Adding an aggregation is easy if it is applied in the end to the set of answers that satisfy the goal, since it requires just a single additional step in the resulting SecreC program. We could treat each intermediate aggregation as an output of a separate SecreC program, and use the result of aggregation as an EDB table in the embedding program. The implementation currently supports only a single aggregation in the goal.

\paragraph{User Interaction.}
In expert systems, the answer is often generated based on user's responses to certain yes/no questions. Formally, we assume that there is a special built-in predicate \texttt{query(X)} that shows to the user the question \texttt{X} and expects a yes/no response.  Since Sharemind does not support live interaction with the user, we split the program into two parts. First of all, the user locally gives answers to all possible questions that the system may ask. The answers comprise a vector of boolean values. When a Sharemind application is executed, these answers are uploaded by the user in a secret-shared manner, and treated as private inputs.

\paragraph{Making use of primary keys.}\label{sec:pk}
If an EDB predicate occurs in the same rule several times, then we need to include the corresponding table several times in the cross product. In some cases, this leads to unnecessary overhead, both in memory and in computation. For example, if we know that the ship name is unique, then a conjunction of literals \texttt{ship(Name,X1,Y1,\_,\_)} and \texttt{ship(Name,\_,\_,Type,Amount)} can be substituted with \texttt{ship(Name,X1,X2,Type,Amount)}. To discover such cases, we allow to mark one column of  an EDB predicate as primary key. Whenever there are two EDB predicates \texttt{p(PK,X1,\ldots,Xn)} and \texttt{p(PK,Y1,\ldots,Yn)} with the same primary key \texttt{PK} in a rule body, only the first occurrence is left behind, and the second one is substituted with \texttt{(Y1=X1) \lpand \ldots \lpand (Yn=Xn)}.

\paragraph{Discovering inconsistencies.}\label{sec:z3}
The unfolding process of Section~\ref{sec:unfold} generates a lot of new rules. These rules in general cannot be immediately evaluated, since they depend on private data that is not known in advance. However, in some cases, we may still discover inconsistencies even without having all the data. For example, if \texttt{X} is a private input, we do not know its exact value, but if the rule body contains conditions \texttt{X > 2} and \texttt{X < 1}, then we can immediately understand that it never evaluates to true. In our implementation, we collect all arithmetic expressions of a rule body and transform these to an input for Z3 solver~\cite{Z3solver}. Some expressions are not supported by the solver, and EDB calls have no meaning for Z3 as well, but as far as inconsistency is found in a subset of rule premises, the entire rule can be treated as false.

\paragraph{Public filters before private filters.}
While the resulting SecreC program formally does not do any filtering and just keeps track of a filter vector \tm{b} which is applied once in the end, in practice we nevertheless apply immediate filtering if the condition is public, which may greatly improve efficiency. Moreover, the formulae $q_i$ of PrivaLog rule body $q_1\ \lpand \ldots \lpand q_n$ that do not contain free variables may be reordered, and public comparisons pushed in front of private comparisons, so that the amount of potential solutions would be reduced as fast as possible.

\subsection{Evaluation}

Our solution is based on 3-party protocol set of Sharemind, which means that we needed to set up 3 servers that would run MPC protocols by exchanging messages in the network.
We benchmarked our implementation on three $2 \times$~Intel Xeon E5-2640 v3 $2.6$ GHz/$8$GT/$20$M servers, with $125$GB RAM running on a LAN throttled down to $100$Mbps, with $40$ms latency. In Table~\ref{tbl:bench}, we present results for simpler programs to give an impression of how large datasets our solution is capable to process so far. The \emph{compilation time} is the running time of the Haskell preprocessor that generates a SecreC program for Sharemind. The \emph{computation time} is the time of executing that program on Sharemind using particular data.

\paragraph{Ship arrival time (\texttt{ship\_mintime.pl})} This is the running example of Section~\ref{sec:example}, extended by a \tm{min} aggregation over time in the goal statement.
\begin{small}
\begin{verbatim}
?-min(arrival(_,portname : private str, cargotype : private int, Time),
      Time, MinTime).
\end{verbatim}
\end{small}

The size of table \texttt{port} is fixed to 5, and the ``input size'' denotes the size of the \texttt{ship} table. The compilation time does not depend on the private data, so it is the same for any number of inputs. The computation times grow linearly in the number of rows in a single table. Computing the minimum is a relatively expensive operation, which takes more time.

\paragraph{Ship delivered cargo (\texttt{ship\_sumcargo.pl})} Let us show how a similar example would work with a sum aggregation, which is much simpler in terms of secure MPC, as it does not require any network communication. Let us state the query as ``how much cargo of certain type arrives at a certain port within certain time?''. We need to slightly modify the rule for \texttt{arrival} and the goal.

\begin{small}
\begin{verbatim}
arrival(Ship,Port,CargoType,Amount,TimeLimit) :-
    ship(Ship,_,_,_,CargoType,Amount),
    suitable_port(Ship,Port),
    reachability_time(Ship,Port,Time),
    Time < TimeLimit.

?-sum(arrival(_,portname : priv str, cargotype : priv int,
              Amount, timelimit : priv int
             ), Amount, SumCargo).
\end{verbatim}
\end{small}

The time of computing the aggregation is now negligible. 

The most time is currently spent on uploading data to Sharemind platform. The main reason is that the current solution computes CRC32 hash of each string one by one, so for $10000$ rows there will be $10000$ rounds of network communication. There are no theoretical issues, and the only problem is that Sharemind does not have a ready API for computing CRC32 of multiple strings as SIMD operation.

\paragraph{Fibonacci numbers (\texttt{fib.pl})} To demonstrate the tool from another point of view, we will try to compile a simple computation of Fibonacci numbers. While computing such numbers in privacy-preserving way is unlikely useful in practice, it demonstrates well some capabilities of our translator.

\begin{small}
\begin{verbatim}
fib(0, 1).    fib(1, 1).

fib(N, F) :-
        N > 1,
        N1 is N-1,
        N2 is N-2,
        fib(N1, F1),
        fib(N2, F2),
        F is F1+F2.

?-fib(X,Y).
\end{verbatim}
\end{small}

This example does not use any EDB predicates at all. Instead, it evaluates Fibonacci number on a private index. The problem here is that we need to know an upper bound on the index in order to know how many Fibonacci numbers to compute. The ``input size'' parameter in this experiment denotes this upper bound.

While we know that computation of Fibonacci numbers is linear, it is not obvious from the rules of \texttt{fib.pl} since we even do not have the assumption that it is defined for natural numbers. Instead, the translator starts constructing facts using bottom-up strategy. In particular, it tries to fit all possible existing facts \texttt{fib} n place of \texttt{fib(N1, F1)} and \texttt{fib(N2, F2)}, and then uses to discovery of inconsistencies (described in Section~\ref{sec:z3}) to get rid of the false facts generated in this way. In this particular example, the size of the new facts does not grow due to simplification, and all generated facts are of the form \texttt{fib(N, F) :- N = A, F = B} for certain constants \texttt{A} and \texttt{B}. Since generation of each new fact is quadratic in the number of ground facts, and $n$ steps are needed to generate $n$ facts, the overall complexity of transformation is cubic, which can also be observed from the translation times of Table~\ref{tbl:bench}. Finally, since all generated facts have exactly the same structure, they are merged into one. For example, the rule that is able to compute the first $n=5$ Fibonacci numbers, as generated by the translator, looks like
\begin{small}
\begin{verbatim}
fib_bf(X_0, X_1) :-
   X_1 = 1 , (X_0 = 0 ; X_0 = 1) ;
   X_1 = 2 , X_0 = 2 ;
   X_1 = 3 , X_0 = 3 ;
   X_1 = 5 , X_0 = 4 .
\end{verbatim}
\end{small}
In order to get a SecreC program that is able to compute any Fibonacci number from $0$ to $n$, we will need to run the translator with $n$ iterations. Technically, the translation pre-computes all $n$ numbers, and the resulting SecreC program chooses one of them obliviously according to the private index. The execution time of SecreC code is thus linear w.r.t the number of choices.

\paragraph{Expert Systems (\texttt{medical.pl}, \texttt{os.pl})}
We have analyzed some simple expert systems written in logic programming languages that we discovered among public GitHub projects. In particular, we analyzed a program for determining operation system fault \texttt{os.pl}\footnote{\url{https://github.com/Anniepoo/prolog-examples/blob/master/expertsystem.pl}} and a simple medical system \footnote{\url{https://github.com/sjbushra/Medical-Diagnosis-system-using-Prolog/blob/master/medical-diagnosis.pl}}. These interactive systems make a decision based on user's responses to certain questions. A potential use case of such privacy-preserving applications would be when a client wants to get advice from an expert system located in an external server, and the client wants to keep their answers private.

We have adapted the syntax of these programs to PrivaLog. This means that we have replaced explicit interaction with the user (by means of predicates \texttt{read} and \texttt{write}) with the special predicate \texttt{query} which makes SecreC program expect the answers to these queries as a private input. Both programs follow a similar pattern, where a hypothesis (a particular disease or a system failure) follows from positive answers to certain questions. The PrivaLog programs look like follows:

\begin{small}
\begin{verbatim}
%medical.pl
symptom(fever) :-
    query('Does patient have a fever (y/n) ?').

symptom(rash) :-
    query('Does patient have a rash (y/n) ?').
.....
hypothesis(mumps) :-
    symptom(fever),
    symptom(swollen_glands).
.....
?-hypothesis(Disease).
\end{verbatim}
\end{small}

\begin{small}
\begin{verbatim}
%os.pl
problem(disc_format) :-
    query('Does the computer show error cannot format').

problem(boot_failure) :-
    query('Does the computer show boot failure').
.....
fault(motherboard) :-
    problem(long_beep),
    problem(short_beep).
.....
?-fault(Fault).
\end{verbatim}
\end{small}
The program \texttt{os.pl} has 18 queries, which are called 26 times. The program \texttt{medical.pl} has 11 queries, which are called 28 queries. From table~\ref{tbl:bench}, we see that the translation time is larger for \texttt{os.pl}. This is because the translator is trying to match each query with each call, so the time is proportional to $18\cdot{26}$ for \texttt{os.pl}, which is around $1.5$ times larger than $11\cdot{28}$ for \texttt{medical.pl}.

\begin{table*}[t]
\centering
\caption{Running times of different examples of logic programs}\label{tbl:bench}
\begin{footnotesize}
\begin{tabular}{| l | r | r | r | r | r | r | r |}
\hline
program & input & unfolding & translation & \multicolumn{4}{c|}{computation time (s)} \\
\cline{5-8}
 & size & strategy & time (s) & readDB & find      & remove     & aggregate \\
 &      &          &          &        & solutions & duplicates &           \\
\hline
upload ship and port data & 100 K & -- & 2510 & -- & -- & -- & -- \\
\hline
\texttt{ship\_mintime.pl} & 10    & any & $< 1$ & 2.8  & 32.3 & 31.2 & 10.2 \\
\texttt{ship\_mintime.pl} & 100   & any & $< 1$ & 2.9  & 33.7 & 52.7 & 12.4 \\
\texttt{ship\_mintime.pl} & 1000  & any & $< 1$ & 4.1  & 37.6 & 71.5 & 15.6 \\
\texttt{ship\_mintime.pl} & 10000 & any & $< 1$ & 51.1 & 92.8 & 118  & 26.9 \\
\hline
\texttt{ship\_sumcargo.pl} & 10    & any & $< 1$ & 2.8  & 34.2 & 29.6 & 0.28 \\
\texttt{ship\_sumcargo.pl} & 100   & any & $< 1$ & 2.9  & 34.9 & 41.0 & 0.28 \\
\texttt{ship\_sumcargo.pl} & 1000  & any & $< 1$ & 4.1  & 38.8 & 61.0 & 0.29 \\
\texttt{ship\_sumcargo.pl} & 10000 & any & $< 1$ & 50.5 & 95.5 & 93.3 & 0.30 \\
\hline
\texttt{fib.pl} & 1 & gr & 0.17 & -- & 3.7 & -- & -- \\
\texttt{fib.pl} & 2 & gr & 0.18 & -- & 3.7  & -- & -- \\
\texttt{fib.pl} & 3 & gr & 0.42 & -- & 8.4 & -- & -- \\
\texttt{fib.pl} & 4 & gr & 0.79  & -- & 9.5 & -- & -- \\
\texttt{fib.pl} & 5 & gr & 1.30  & -- & 10.3 & -- & -- \\
\texttt{fib.pl} & 6 & gr & 1.99  & -- & 11.2 & -- & -- \\
\texttt{fib.pl} & 7 & gr & 2.87 & -- & 12.5 & -- & -- \\
\texttt{fib.pl} & 8 & gr & 3.99 & -- & 12.6 & -- & -- \\
\texttt{fib.pl} & 9 & gr & 5.58 & -- & 13.2 & -- & -- \\
\texttt{fib.pl} & 10 & gr & 7.21  & -- & 13.9 & -- & -- \\
\texttt{fib.pl} & 15 & gr & 22.0  & -- & 16.3 & -- & -- \\
\texttt{fib.pl} & 20 & gr & 46.0  & -- & 20.6 & -- & -- \\
\texttt{fib.pl} & 100 & gr & 6000  & -- & 50.5 & -- & -- \\
\hline
\texttt{medical.pl} & 1  & bfs & 10.5 & -- & 16.2 & -- & -- \\
\hline
\texttt{os.pl} & 1  & bfs & 15.0 & -- & 18.9 & -- & -- \\
\hline
\end{tabular}
\end{footnotesize}
\end{table*}

\section{Conclusion}

In this paper, we have presented PrivaLog, which is a logic programming language with data privacy types. The language is accompanied with a translator to SecreC language, which allows to execute the program on Sharemind system using secure multi-party-computation. We have proven that the transformation preserves program semantics, and that the data labeled as ``private'' will remain confidential during the computation. We have provided a transformation tool that makes use of several optimizations, and evaluated it on some examples.

\bibliographystyle{plain}
\bibliography{references}

\appendix

\end{document}